\documentclass[10pt]{article}
\usepackage{authblk}
\usepackage[nocompress]{cite}
\usepackage{tikz}
\usepackage{amsmath}
\usepackage{amsthm}
\usepackage{mathtools}
\usepackage{enumerate}
\usepackage{subcaption}
\usepackage{hyperref}
\usepackage{graphicx}
\usepackage{algorithm2e}
\usepackage[a4paper, total={5.5in, 8in}]{geometry}
\usepackage{xspace}
\usepackage{xspace}
\usetikzlibrary{automata,arrows,shapes,decorations.pathreplacing,patterns}
\tikzstyle{every node}=[circle,draw=black, inner sep=1.5pt,fill=white]
\usepackage{thmtools,thm-restate}
\usepackage{cleveref}

\DeclareMathOperator*{\argmin}{arg\,min}

\newcommand{\ltdots}{..}

\newcommand\BB{\ensuremath{\mathtt{b}}\xspace} 
\newcommand\E{\ensuremath{\mathtt{e}}\xspace}
\newcommand\zero{\ensuremath{\mathtt{0}}\xspace}
\newcommand\one{\ensuremath{\mathtt{1}}\xspace}
\newcommand\ov{\textsf{OV}\xspace}
\newcommand\msa{\textsf{MSA}\xspace}
\newcommand\msas{\textsf{MSAs}\xspace}
\newcommand\msamn{\textsf{MSA}$[1\ltdots m,1\ltdots n]$\xspace}
\newcommand\msaij[2]{\mathsf{MSA}[#1,#2]\xspace}
\newcommand\eds{\textsf{EDS}\xspace}
\newcommand\edses{\textsf{EDS}s\xspace}
\newcommand\efg{\textsf{EFG}\xspace}
\newcommand\efgs{\textsf{EFG}s\xspace}
\newcommand\ovh{\textsf{OVH}\xspace}
\newcommand\spell{\mathsf{spell}\xspace}

\newtheorem{theorem}{Theorem}
\newtheorem{lemma}{Lemma}
\newtheorem{proposition}{Proposition}
\newtheorem{corollary}{Corollary}
\newtheorem{observation}{Observation}

\newtheorem{definition}{Definition}

\begin{document}

\title{Algorithms and Complexity on Indexing \\Founder Graphs\thanks{This is a preprint of the extended full version of conference papers in WABI 2020 \cite{MCENT20} and ISAAC 2021 \cite{ENACTM21}. An earlier version of this preprint has also been reviewed and presented at RECOMB-seq 2021 highlights track.}}%

\author[1]{Massimo Equi}
\author[1]{Tuukka Norri}
\author[1,2]{Jarno Alanko}
\author[3]{Bastien Cazaux}
\author[1]{Alexandru I. Tomescu}
\author[1,$\star$]{Veli M\"akinen}

\affil[1]{Department of Computer Science, University of Helsinki, Finland.}
\affil[2]{Faculty of Computer Science, Dalhousie University, Canada}
\affil[3]{LIRMM, Univ. Montpellier, CNRS, France}
\affil[$\star$]{Contact: veli.makinen@helsinki.fi}

\setcounter{Maxaffil}{0}
\renewcommand\Affilfont{\itshape\small}

\maketitle
\thispagestyle{empty}

\begin{abstract}
We study the problem of matching a string in a labeled graph. Previous research has shown that unless the \emph{Orthogonal Vectors Hypothesis} (OVH) is false, one cannot solve this problem in strongly sub-quadratic time, nor index the graph in polynomial time to answer queries efficiently (Equi et al. ICALP 2019, SOFSEM 2021). These conditional lower-bounds cover even deterministic graphs with binary alphabet, but there naturally exist also graph classes that are easy to index: For example, \emph{Wheeler graphs} (Gagie et al.~\emph{Theor. Comp. Sci.} 2017) cover graphs admitting a Burrows-Wheeler transform -based indexing scheme. However, it is NP-complete to recognize if a graph is a Wheeler graph (Gibney, Thankachan, ESA 2019). 

We propose an approach to alleviate the construction bottleneck of Wheeler graphs. Rather than starting from an arbitrary graph, we study graphs induced from \emph{multiple sequence alignments} (\msas). \emph{Elastic degenerate strings} (Bernadini et al. SPIRE 2017, ICALP 2019) can be seen as such graphs, and we introduce here their generalization: \emph{elastic founder graphs}. We first prove that even such induced graphs are hard to index under OVH. Then we introduce two subclasses, repeat-free and semi-repeat-free graphs, that are easy to index. We give a linear time algorithm to construct a repeat-free non-elastic founder graph from a gapless \msa, and (parameterized) near-linear time algorithms to construct semi-repeat-free (repeat-free, respectively) elastic founder graphs from general \msas. Finally, we show that repeat-free elastic founder graphs admit a reduction to Wheeler graphs in polynomial time. 
\end{abstract}

\section{Introduction\label{sec:introduction}}

In string research, many different problems relate to the common question of how to handle a collection of strings. When such a collection contains very similar strings, it can be represented as some ``high scoring'' \emph{Multiple Sequence Alignment} (\msa), i.e., as a matrix \msamn whose $m$ rows are the individual strings each of length $n$, which may include special ``gap'' symbols such that the columns represent the aligned positions. While it is NP-hard to find an optimal \msa even under the simplest score of maximizing the number of identity columns (i.e., longest common subsequence length)~\cite{Mai78}, the central role of \msa as a model of biological evolution has resulted into numerous heuristics to solve this problem in practice~\cite{Chaetal15}. In this paper, we assume an \msa as an input.

A simple way to define the problem of finding a match for a given string in the \msa is to ask whether the string matches a substring of some row (ignoring gap symbols). This leads to the widely studied problem of indexing repetitive text collections, see, e.g., references~\cite{MNSV09jcb,Naetal13a,Naetal13b,Naetal16a,Naetal16b,GN19,GNP20}. These approaches reducing an \msa to plain text reach algorithms with linear time complexities. However, the performance of these algorithms rely on the fact that a match is always found within an individual row of the \msa.

One feature worth considering is the possibility to allow a match to jump from any row to any other row of the \msa between consecutive columns. This property is usually referred to as \emph{recombination} due to its connection to evolution. To solve this version of the problem, different approaches have to be used, and a possible alternative is a graph representation of the \msa. Figure~\ref{fig:overview-vg} shows a simple solution, which consists in turning distinct characters of each column into nodes, and then adding the edges supported by row-wise connections. In this graph, a path whose concatenation of node labels matches a given string represents a match in the original \msa (ignoring gaps). Refinements of this approach are mostly used in bioinformatics~\cite{marschall2016computational}, where recombination is a desired feature, and it is realized by the fact that the resulting graph encodes a super-set of the strings of the original \msa. 

Aligning a sequence against a graph is not a trivial task. Only quadratic solutions are known~\cite{AmirLL00,manber1992approximate, rautiainen2017aligning}, and this was recently proved to be a conditional lower bound for the problem~\cite{EGMT19}. Moreover, even attempting to index the graph to query the string faster presents significant difficulties. On one hand, indexes constructed in polynomial time still require quadratic-time queries in the worst case~\cite{Tha13}. On the other hand, worst-case linear-time queries are possible, but this has the potential to make the index grow exponentially~\cite{SVM14}. These might be the best results possible for general graphs and DAGs without any specific structural property, as the need for exponential indexing time to achieve sub-quadratic time queries constitutes another conditional lower bound for the problem~\cite{EMT21}.

Thus, if we want to achieve better performances, we have to make more assumptions on the structure of the input, so that the problem might become tractable. Following this line, a possible solution consists in identifying special classes of graphs 
that, while still able to represent any \msa, have a more limited amount of recombination, thus allowing for fast matching or fast indexing. This is the case for \emph{Elastic Degenerate Strings} (\edses) \cite{aoyama_et_alCPM2018,bernardini_et_al2019elastic,Beretal20,Beretal17,IKP17}, which can represent an \msa as a sequence of sets of strings, in which a match can span consecutive sets, using any one string in each of these (see Figure~\ref{fig:overview-eds-efg}, graph in the center). The advantage of this structure is that it is possible to perform expected-case subquadratic time queries \cite{bernardini_et_al2019elastic}. However, \edses are still hard to index~\cite{GibneySPIRE2020}, and there is a lack of results on how to derive a ``suitable'' \eds from an \msa.

In this context, we propose a generalization of an \eds to what we call an \emph{Elastic Founder Graph} (\efg). An \efg is a DAG that, as an \eds, represents an \msa as a sequence of sets of strings; each set is called a \emph{block}, and each string inside a block is represented as a labeled node. The difference with \edses is that the nodes of two consecutive blocks are not forced to be fully connected
. This means that, while in an \eds a match can always pair any string of a set with any string of the next set, in an \efg it might be the case that only some of these pairings are allowed. Figure~\ref{fig:overview-eds-efg} illustrates these differences. Allowing for more selective connectivity between consecutive blocks also means that finding a match for a string in an \efg is harder than in an \eds. This is because \edses are a special case of \efgs, hence the hardness results for the former carry to the latter. Specifically, a previous work \cite{GT19} showed that, under the \emph{Orthogonal Vectors Hypothesis} (\ovh), no index for \edses constructed in polynomial time can provide queries in time $O(|Q|+|\widetilde{T}|^\delta |Q|^\beta)$, where $|\widetilde{T}|$ is the number of sets of strings, $|Q|$ is the length of the pattern and $\beta<1$ or $\delta<1$. Nevertheless, in this work we present an even tighter quadratic lower bound for \efgs, proving that, under \ovh, an index built in time $O(|E|^\alpha)$ cannot provide queries in time $O(|Q|+|E|^\delta|Q|^\beta)$, where $|E|$ is the number of edges and $\beta<1$ or $\delta<1$. Notice that $|\widetilde{T}|$ could even be $o(|E|)$ (e.g. an \efg of two fully connected blocks), hence our lower bound more closely relates to the total size of an \efg. Additionally, the earlier lower bound  \cite{GT19} naturally applies only to indexing \edses, and is obtained by performing many hypothetical fast queries; ours is derived by first proving a quadratic \ovh-based lower bound for the \emph{online} string matching problem in \efgs, and then using a general result \cite{EMT21} to simply translate this into an indexing lower bound.

\begin{figure*}[t]
    \begin{subfigure}[b]{\linewidth}
    \includegraphics[scale=0.25]{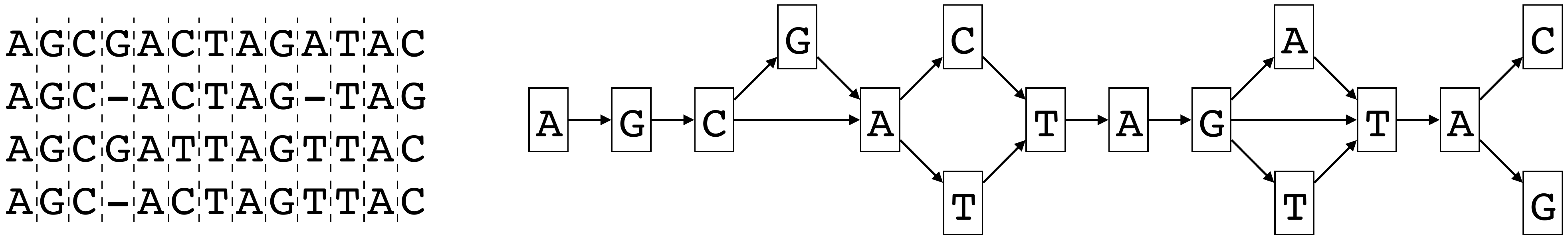}   
    \caption{A column-by-column segmentation of an \msa on the left, leading to the variation graph on the right.   \label{fig:overview-vg}}
    \end{subfigure}
    \smallskip
    \begin{subfigure}[b]{\linewidth}
    \includegraphics[scale=0.25]{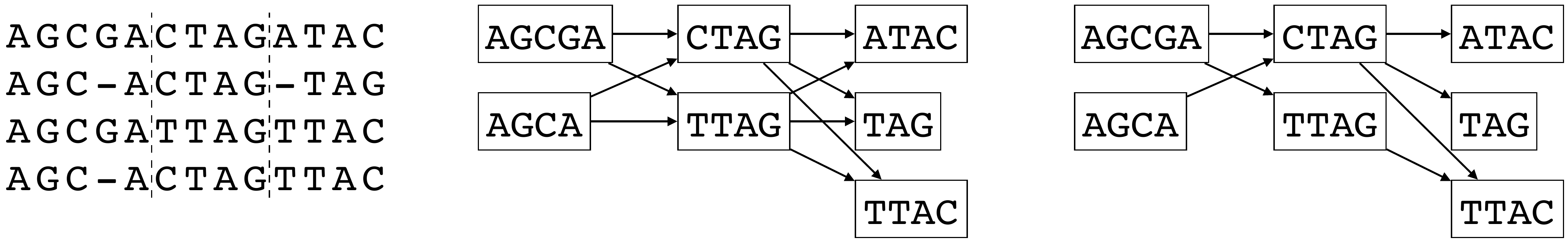} 
    \caption{A different segmentation of the \msa, leading to the \eds in the center, and the \efg on the right. Notice that in an \eds every node is connected with all nodes to the right, while in an \efg edges are added only if their endpoints are consecutive in some row of the \msa (as in the case of variation graphs).\label{fig:overview-eds-efg}}
    \end{subfigure}
    \smallskip
    \begin{subfigure}[b]{\linewidth}
    \includegraphics[scale=0.25]{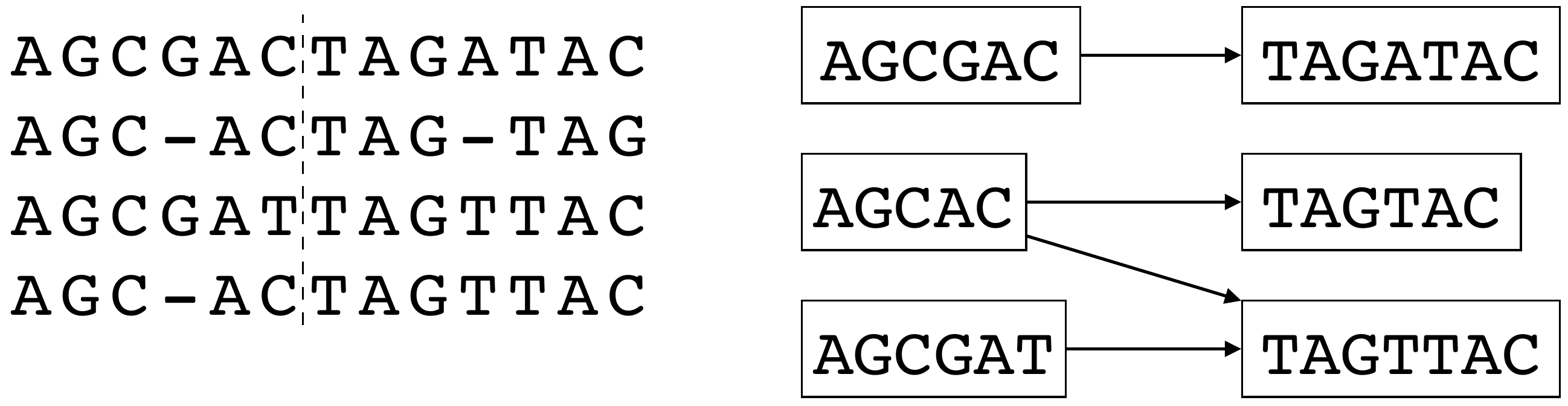} 
    \caption{A segmentation of the \msa that leads to a repeat-free \efg (i.e.~no node label has another occurrence on some path of the \efg).\label{fig:overview-rf-efg}}
    \end{subfigure}
    \smallskip
    \begin{subfigure}[b]{\linewidth}
    \includegraphics[scale=0.25]{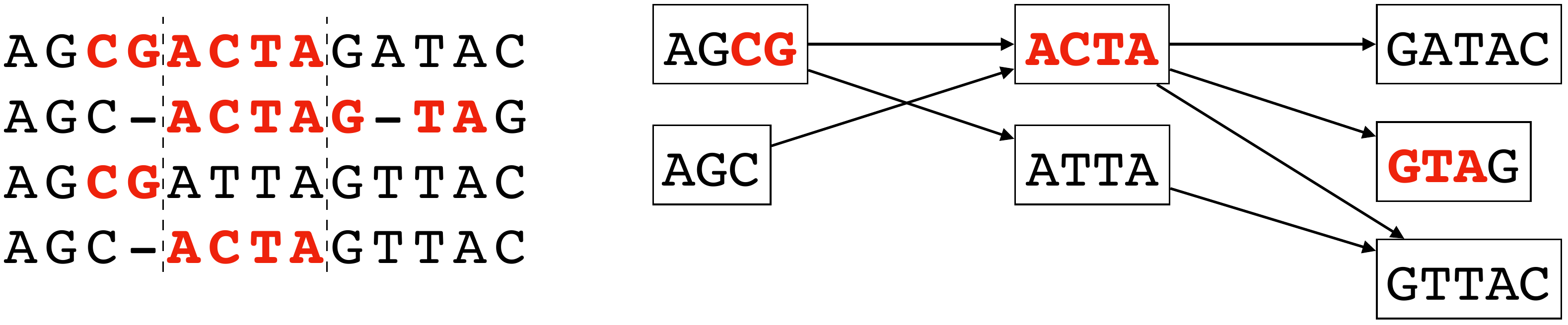} 
    \caption{A segmentation of the \msa that leads to a semi-repeat-free \efg (i.e.~no node label has another occurrence on some path of the \efg, except as a prefix of another node in the same segment). An occurrence of query $Q=\mathtt{CGACTAGTA}$ in \efg is depicted in red. As can be seen, such query does not have an occurrence in a single row of the \msa. \label{fig:overview-srf-efg}}
    \end{subfigure}
    \caption{An \msa on the left, and various graph-based representations of it on the right. Notice that in all graphs (except the \eds) edges are added only between nodes that are observed as consecutive in some row of the \msa.\label{fig:MSA_Repr}}
\end{figure*}

Then, in order to break through these lower bounds, we identify two natural classes of \efgs, which respect what we call \emph{repeat-free} and \emph{semi-repeat-free} properties. The repeat-free property (Figure~\ref{fig:overview-rf-efg}) forces each string in each block to occur only once in the entire graph, and the semi-repeat-free property (Figure~\ref{fig:overview-srf-efg}) is a weaker form of this requirement. Thanks to these properties, we can more easily locate substrings of a query string in repeat-free \efgs and semi-repeat-free \efgs. In particular, (semi-)repeat-free \efgs and \edses can be indexed in polynomial time for linear time string matching.

One might think that these time speedups come with a significant cost in terms of flexibility. Instead, the special structure of these \efgs do not hinder their expressive power. Indeed, we show that an \msa can be ``optimally'' segmented into blocks inducing a repeat-free or semi-repeat-free \efg. Clearly, this depends on how one chooses to define optimality. We consider three optimality notions: maximum number of blocks, minimum maximum block height, and minimum maximum block length. In Figure~\ref{fig:overview-srf-efg}, the first score is 3, second is 3, and the third is 5. The two latter notions stem from the earlier work on segmentations~\cite{NCKM19,CKMN19}, now combined with the (semi)-repeat-free constraint. The first is the simplest optimality notion, now making sense combined with the (semi)-repeat-free constraint. 

For each of these optimality notions, we give a polynomial-time dynamic programming algorithm that converts an \msa into an optimal (semi-)repeat-free \efg if such exists. For the first and the third notion combined with the semi-repeat-free constraint, we derive more involved solutions with almost optimal $O(mn \log m)$ and $O(mn \log m+n\log \log n)$ running time, respectively. Futhermore, we give an (optimal) $O(mn)$ time solution for the special case of \msa without gap symbols. The algorithm for the special case uses a monotonicity property not holding with gaps. With general \msas we delve into the combinatorial properties of repetitive string collections synchronized with gaps and show how to use string data structures in this setting. The techniques can be easily adapted for other notions of optimality.

Another class of graphs that admits efficient indexing are Wheeler graphs~\cite{GMS17}, which offer an alternative way to model an \efg and thus a \msa. However, it is NP-complete to recognize if a given graph is a Wheeler graph~\cite{GT19}, and thus, to use the efficient algorithmic machinery around Wheeler graphs~\cite{Alaetal20} one needs to limit the focus on indexable graphs that admit efficient construction. Indeed, we show that any \efg that respects the repeat-free property can be reduced to a Wheeler graph in polynomial time.
Interestingly, we were not able to modify this reduction to cover the semi-repeat-free case, leaving it open if these two notions of graph indexability have indeed different expressive power, and whether there are more graph classes with distinctive properties in this context.

The paper is structured as follows. At the high level, we first focus on gapless \msas, and then we extend the results to the general case. In more detail, Section~ \ref{sect:definitions} defines the founder graph concepts and explores some basic techniques. Section~\ref{sect:repeat-freeness} gives the first indexing results using just classical data structures as a warm up. Section~\ref{sect:construction} covers linear time construction of repeat-free (non-elastic) founder graphs from gapless \msas. Section~\ref{sect:compressedindexing} improves over the basic indexing results using succinct data structures. Section~\ref{sect:hardness} gives the proof of conditional indexing hardness when moving from the gapless case to the general case of \efgs. Indexing results are generalized to (semi-)repeat-free \efgs in Section~\ref{sect:gaps}.
Construction results are generalized to (semi-)repeat-free \efgs in Section~\ref{sect:efgconstruction}. Connection to Wheeler graphs is considered in Section~\ref{sect:wheeler}. Implementation is discussed in Section~\ref{sect:experiments}. Finally, future directions are discussed in Section~\ref{sect:discussion}.

\section{Definitions and basic tools\label{sect:definitions}}

\subsection{Strings}

We denote integer intervals by $[i..j]$. Let $\Sigma = \{1, \ldots, \sigma\}$ be an alphabet of size $|\Sigma| = \sigma$.
A \emph{string} $T[1..n]$ is a sequence of symbols from $\Sigma$, i.e. $T\in \Sigma^n$, where 
$\Sigma^n$ denotes the set of strings of length $n$ under the alphabet $\Sigma$.
A \emph{suffix} of string $T[1..n]$ is $T[i..n]$ for $1\leq i\leq n$. 
A \emph{prefix} of string $T[1..n]$ is $T[1..i]$ for $1\leq i\leq n$. 
A \emph{substring} of string $T[1..n]$ is $T[i..j]$ for $1\leq i\leq j \leq n$. 
The \emph{length} of a string $T$ is denoted $|T|$. The \emph{empty string} is the string of length $0$.
In particular, substring $T[i..j]$ where $j<i$ is the empty string. The \emph{lexicographic order} of two strings $A$ and $B$ is naturally defined by the order of the alphabet: $A<B$ iff $A[1..i]=B[1..i]$ and $A[i+1]<B[i+1]$ for some $i\geq 0$. If $i+1>\min(|A|,|B|)$, then the shorter one is regarded as smaller. However, we usually avoid this implicit comparison by adding \emph{end marker} $\mathbf{0}$ to the strings.
Concatenation of strings $A$ and $B$ is denoted $AB$.

\subsection{Elastic founder graphs}

As mentioned in the introduction, our goal is to compactly represent an \msa using an elastic founder graph. In this section we formalize these concepts.

A \emph{multiple sequence alignment} \msamn is a matrix with $m$ strings drawn from $\Sigma \cup \{\text{-}\}$, each of length $n$, as its rows. Here $\text{-} \notin \Sigma$ is the \emph{gap} symbol. For a string $X \in \left(\Sigma \cup \{\text{-}\}\right)^*$, we denote $\spell(X)$ the string resulting from removing the gap symbols from $X$.

Let $\mathcal{P}$ be a \emph{partitioning} of $[1..n]$, that is, a sequence of subintervals $\mathcal{P}=[x_1..y_1]$, $[x_2..y_2],\ldots,[x_b..y_b]$, where $x_1=1$, $y_b=n$, and for all $j>2$, $x_j=y_{j-1}+1$. A \emph{segmentation} $S$ of \msamn based on partitioning $\mathcal{P}$ is a sequence of $b$ sets $S^k= \{\spell(\msaij{i}{x_k..y_k}) \mid 1\leq i\leq m\}$ for $1\leq k\leq b$; in addition, we require for a (proper) segmentation that $\spell(\msaij{i}{x_k..y_k})$ is not an empty string for any $i$ and $k$. We call set $S^k$ a \emph{block}, while $\msaij{1..m}{x_k..y_k}$ or just $[x_k..y_k]$ is called a \emph{segment}. The \emph{length} of block $S^k$ is $L(S^k)=y_k-x_k+1$ and the \emph{height} of block $S^k$ is $H(S^k)=|S^k|$. Segmentation naturally leads to the definition of a founder graph through the block graph concept:

\begin{definition}[Block Graph]
A \emph{block graph} is a graph $G=(V,E,\ell)$ where $\ell: V \rightarrow \Sigma^+$ is a function that assigns a string label to every node and for which the following properties hold.
\begin{enumerate}
	\item Set $V$ can be partitioned into a sequence of $b$ \emph{blocks} $V^1, V^2, \ldots, V^b$, that is, $V = V^1 \cup V^2 \cup \cdots \cup V^b$ and $V^i \cap V^j = \emptyset$ for all $i\neq j$;
	\item If $(v,w) \in E$ then $v \in V^i$ and $w \in V^{i+1}$ for some $1 \leq i \leq b-1$; and
	\item if $v,w \in V^i$ then $|\ell(v)| = |\ell(w)|$ for each $1 \leq i \leq b$ and if $v\neq w$, $\ell(v) \neq \ell(w)$.
\end{enumerate}
\end{definition}

With \emph{gapless} \msas, block $S^k$ equals segment $\msaij{1..m}{x_k..y_k}$, and in that case the \emph{founder graph} is a block graph induced by segmentation $S$. The idea is to have a graph in which the nodes represent the strings in $S$ while the edges retain the information of how such strings can be recombined to spell any sequence in the original \msa. With general \msas with gaps, we consider the following extension, with an analogy to \edses~\cite{bernardini_et_al2019elastic}:   

\begin{definition}[Elastic block and founder graphs]
We call a block graph \emph{elastic} if its third condition is relaxed in the sense that each $V^i$ can contain non-empty variable-length strings. An \emph{elastic founder graph} (\efg) is an elastic block graph $G(S) = (V,E,\ell)$ \emph{induced} by a segmentation $S$ as follows: For each $1 \leq k \leq b$ we have $S^k = \{\spell(\msaij{i}{x_k..y_k}) \mid 1\leq i\leq m\} = \{\ell(v) : v \in V^k\}$. It holds $(v,w) \in E$ if and only if there exists $k \in [1 \ltdots b-1]$ and $t \in [1 \ltdots m]$ such that $v \in V^k$, $w \in V^{k+1}$ and $\spell(\msaij{t}{x_k..y_{k+1}})= \ell(v)\ell(w)$.
\end{definition}

By definition, (elastic) founder and block graphs are acyclic. For convention, we interpret the direction of the edges as going from left to right. Consider a path $P$ in $G(S)$ between any two nodes. The label $\ell(P)$ of $P$ is the concatenation of labels of the nodes in the path. Let $Q$ be a query string. We say that $Q$ \emph{occurs} in $G(S)$ if $Q$ is a substring of $\ell(P)$ for any path $P$ of $G(S)$.
Figure~\ref{fig:MSA_Repr} illustrates such a query.

As we later learn, some further properties on founder graphs are needed for supporting fast queries:
\begin{definition}
\efg $G(S)$ is \emph{repeat-free} if each $\ell(v)$ for $v\in V$ occurs in $G(S)$ only as prefix of paths starting with $v$. 
\end{definition}

We also consider a variant that is relevant due to variable-length strings in the blocks:
\begin{definition}
\efg $G(S)$ is \emph{semi-repeat-free} if each $\ell(v)$ for $v\in V$ occurs in $G(S)$ only as prefix of paths starting with $w\in V$, where $w$ is from the same block as $v$. 
\end{definition}

These definitions also apply to general elastic block graphs and to elastic degenerate strings as their special case.

We note that not all \msas admit a segmentation leading to a (semi-)repeat-free \efg
, e.g.~an alignment with rows \texttt{-A} and \texttt{AA}. However, our algorithms detect such cases, 
thus one can build an \efg consisting of just one block with the rows of the \msa (with gaps removed). Such \efgs can be indexed using standard string data structures to support efficient queries.

\subsection{Basic tools}

A \emph{trie} \cite{Bri59} of a set of strings is a rooted directed tree with outgoing edges of each node labeled by distinct characters such that there is a root to leaf path spelling each string in the set; the shared part of the root to leaf paths to two different leaves spell the common prefix of the corresponding strings. Such a trie can be computed in $O(N \log \sigma)$ time, where $N$ is the total length of the strings, and it supports string queries that require $O(q \log \sigma)$ time, where $q$ is the length of the queried string. 

In a \emph{compact trie} the maximal non-branching paths of a trie become edges labeled with the concatenation of labels on the path. \emph{Suffix tree} is the compact trie of all suffixes of string $T\mathbf{0}$. In this case, the edge labels are substrings of $T$ and can be represented in constant space as an interval. Such tree takes linear space and can be constructed in linear time~\cite{Farach97} so that when reading the leaves from left to right, the suffixes are listed in their lexicographic order. The leaves hence form the \emph{suffix array}~\cite{MM93} of string $T$, which is an array $\mathtt{SA}[1..n+1]$ such that $\mathtt{SA}[i]=j$ if $T'[j..n+1]$ is the $i$-th smallest suffix of string $T'=T\mathbf{0}$, where $T \in \{1,2,\ldots,\sigma\}^n$. A \emph{generalized suffix tree or array} is one built on a set of strings. In this case, string $T$ above is the concatenation of the strings with symbol $\mathbf{0}$ between each.

Let $Q[1..m]$ be a query string. If $Q$ occurs in $T$, then the \emph{locus} or \emph{implicit node} of $Q$ in the suffix tree of $T$ is $(v,k)$ such that $Q=XY$, where $X$ is the path spelled from the root to the parent of $v$ and $Y$ is the prefix of length $k$ of the edge from the parent of $v$ to $v$. The leaves of the subtree rooted at $v$ are then all the suffixes sharing the common prefix $Q$. Let the left- and right-most leaves in that subtree be the $c$-th and $d$-th smallest suffixes of $T\mathbf{0}$. Then $Q=T[i..i+|Q|-1]$ for  $i\in \mathtt{SA}[c..d]$ and does not occur elsewhere. We use heavily this connection of suffix tree node $v$ and suffix array interval $\mathtt{SA}[c..d]$. Moreover, there are succinct data structures to do this mapping in both directions in constant time~\cite{Sad07}. 

Let $aX$ and $X$ be paths spelled from the root of a suffix tree to nodes $v$ and $w$, respectively. Then one can store a \emph{suffix link} from $v$ to $w$. \emph{Implicit suffix links} for implicit nodes are defined analogously, but they are not stored explicitly. In many algorithms, one can simulate implicit suffix links through explicit suffix links, as the work amortizes to a constant per step.

An \emph{Aho-Corasick automaton} \cite{AC75} is a trie of a set of strings with additional pointers (fail-links). While scanning a query string, these pointers (and some shortcut links on them) allow to identify all the positions in the query at which a match for any of the strings occurs. Construction of the automaton takes the same time as that of the trie. Queries take $O(q \log \sigma +\mathtt{occ})$ time, where $\mathtt{occ}$ is the number of matches.

A \emph{Burrows-Wheeler transform} BWT$[1..n+1]$ \cite{BW94} of string $T$ is such that BWT$[i]=T'[\mathtt{SA}[i]-1]$, where $T'=T\mathbf{0}$ and $T'[-1]$ is regarded as $T'[n+1]=\mathbf{0}$.

A \emph{bidirectional BWT index} \cite{SOG12,BCKM20} is a succinct index structure based on some auxiliary data structures on BWT. Given a string $T \in \Sigma^n$, with $\sigma\leq n$, such index occupying $O(n \log \sigma)$ bits of space can be built in randomized $O(n)$ time and it supports finding in $O(q)$ time if a query string $Q[1..q]$ appears as substring of $T$ \cite{BCKM20}. Moreover, the query returns an interval pair  $([i..j]$,$[i'..j'])$ such that suffixes of $T$ starting at positions $\mathtt{SA}[i],\mathtt{SA}[i+1],\ldots, \mathtt{SA}[j]$ share a common prefix matching the query. Interval $[i'..j']$ is the corresponding interval in the suffix array of the reverse of $T$. Let $([i..j]$,$[i'..j'])$ be the interval pair corresponding to query substring $Q[l..r]$. A \emph{bidirectional backward step} updates the interval pair $([i..j]$,$[i'..j'])$ to the corresponding interval pair when the query substring $Q[l..r]$ is extended to the left into $Q[l-1..r]$ (\emph{left extension}) or to the right into $Q[l..r+1]$ (\emph{right extension}). Such step takes constant time \cite{BCKM20}.

A \emph{fully-functional bidirectional BWT index} \cite{BC19} expands the steps to allow contracting symbols from the left or from the right. That is, substring $Q[l..r]$ can be modified into $Q[l+1..r]$ (\emph{left contraction}) or to $Q[l..r-1]$ (\emph{right contraction}) and the the corresponding interval pair can be updated in constant time.

Among the auxiliary structures used in BWT-based indexes, we explicitly use the \emph{rank} and \emph{select} structures: String $B[1..n]$ from binary alphabet is called a \emph{bitvector}. Operation $\mathtt{rank}(B,i)$ returns the number of $1$s in $B[1..i]$. Operation $\mathtt{select}(B,j)$ returns the index $i$ containing the $j$-th $1$ in $B$.  Both queries can be answered in constant time using an index requiring $o(n)$ bits in addition to the bitvector itself \cite{Jac89}.

We summarize a result that we use later.
\begin{lemma}[\cite{BC19}]\label{lemma:randomized_index_construction}
Given a text $T$ of length $n$ from an alphabet  $\{1,2,\ldots, \sigma\}$, it is possible to construct in $O(n)$ randomized time and $O(n \log \sigma)$ bits of space a bidirectional BWT index that supports left extensions and right contractions in $O(1)$ time.
\end{lemma}

\subsubsection{Deterministic index construction for integer alphabets}

We now describe how to replace the randomized linear time construction of the bidirectional BWT index with a deterministic one, so that left extensions and right contractions are supported in $O(\log \sigma)$ time. We need only a single-directional subset of the index of Belazzougui and Cunial \cite{BC19}, consisting of an index on the BWT, augmented with balanced parentheses representations of the topologies of the suffix tree of $T$ and of the suffix link tree of the reverse of $T$, such that the nodes corresponding to maximal repeats in both topologies are marked \cite{BC19}.

Belazzougui et al. showed how to construct the BWT of a string $O(n)$ deterministic time and $O(n \log \sigma)$ bits of space \cite{BCKM20}. Their algorithm can be used to construct both the BWT of $T$ and the BWT of the reverse of $T$ in $O(n)$ time. We can build the bidirectional BWT index \cite{belazzougui2013versatile,SOG12} by indexing both BWTs as \emph{wavelet trees} \cite{GGV03}. The bidirectional index can then be used to construct both of the required tree topologies in $O(n \log \sigma)$ time using bidirectional extension operations and the counter-based topology construction method of Belazzougui et al. \cite{BCKM20}.  See the supplement of a paper on variable order Markov models by Cunial et al. \cite{cunial2019framework} for more details on the construction of the tree topologies. The topologies are then indexed for various navigational operations required by the contraction operation described in \cite{BC19}.

This index enables left extensions in $O(\log \sigma)$ time using the BWT of $T$, and right contractions in constant time using the succinct tree topologies as shown by Belazzougui and Cunial \cite{BC19}. We summarize the result in the lemma below.

\begin{lemma}\label{lemma:index_construction}
Given a text $T$ of length $n$ from an alphabet  $\{1,2,\ldots, \sigma\}$, it is possible to construct in $O(n \log \sigma)$ deterministic time and $O(n \log \sigma)$ bits of space a bidirectional BWT index that supports left extensions in $O(\log \sigma)$ time and and right contractions in $O(1)$ time.
\end{lemma}

We note that it may be possible to improve the time of left extensions to $O(\log \log \sigma)$ by replacing monotone minimum perfect hash functions by slower yet deterministic linear time constructable data structures in all constructions leading to Theorem 6.7 of Belazzougui et al. \cite{BCKM20}.

\section{Indexable repeat-free founder graphs\label{sect:repeat-freeness}}

We now consider non-elastic founder graphs induced from gapless \msas, and later turn back to the general case. We show that there exists a family of founder graphs that admit a polynomial time constructable index structure supporting fast string matching. First, a trivial observation: the input multiple alignment is a founder graph for the segmentation consisting of only one segment. Such founder graph (set of sequences) can be indexed in linear time to support linear time string matching \cite{BCKM20}. Now, the question is, are there other segmentations that allow the resulting founder graph to be indexed in polynomial time? We show that this is the case.

\begin{proposition}
Repeat-free founder graphs can be indexed in polynomial time to support polynomial time string queries.
\label{prop:blockrepeatfree}
\end{proposition}

To prove the proposition, we construct such an index and show how queries can be answered efficiently. Our first solution uses just classical data structures, and works as a warm up: Later we improve this solution using succinct data structures, and while doing so we exploit the connections to the derivations in this section.

Let $P(v)$ be the set of all paths starting from node $v$ and ending in a sink node. Let $P(v,i)$ be the set of \emph{suffix path labels} $\{\ell(L)[i..]\mid L \in P(v)\}$ for $1\leq i \leq |\ell(v)|$. Consider sorting $\mathcal{P}=\cup_{v \in V,1\leq i\leq |\ell(v)|} P(v,i)$ in lexicographic order. Then one can binary search any query string $Q$ in $\mathcal{P}$ to find out if it occurs in $G(S)$ or not. The problem with this approach is that $\mathcal{P}$ is of exponential size. 

However, if we know that $G(S)$ is repeat-free, we know that the lexicographic order of $\ell(L)[i..]$, $L \in P(v)$ in $\mathcal{P}$, is fully determined by the prefix $\ell(v)[i..|\ell(v)|]\ell(w)$ of $\ell(L)[i..]$, where $w$ is the node following $v$ on the path $L$, except against other suffix path labels starting with $\ell(v)[i..|\ell(v)|]\ell(w)$. Let $P'(v,i)$ denote the set of suffix path labels cut in this manner. Now the corresponding set $\mathcal{P}'=\cup_{v \in V,1\leq i\leq |\ell(v)|} P'(v,i)$ is no longer of exponential size. Consider again binary searching a string $Q$ in sorted $\mathcal{P}'$. If $Q$ occurs in $\mathcal{P}'$ then it occurs in $G(S)$. If not, $Q$ has to have some $\ell(v)$ for $v\in V$ as its substring in order to occur in $G(S)$. 

To figure out if $Q$ contains $\ell(v)$ for some $v\in V$ as its substring, we build an Aho-Corasick automaton \cite{AC75} for $\{\ell(v) \mid v \in V\}$. Scanning this automaton takes $O(|Q|\log \sigma)$ time and returns such $v \in V$ if it exists. 

To verify such a potential match, we need several tries \cite{Bri59}. For each $v\in V$, we build tries $\mathcal{R}(v)$ and $\mathcal{F}(v)$ on the sets 
$\{\ell(u)^{-1} \mid (u,v)\in E\}$ and $\{\ell(w) \mid (v,w)\in E\}$, respectively, where $X^{-1}$ denotes the reverse $x_{|X|}x_{|X|-1}\cdots x_1$ of string $X=x_1x_2 \cdots x_{|X|}$.

Assume now we have located (using the Aho-Corasick automaton) $v\in V$ with $\ell(v)$ such that $\ell(v)=Q[i..j]$, where $v$ is at the $k$-th block of $G(S)$. We continue searching $Q[1..i-1]$ from right to left in trie $\mathcal{R}(v)$. If we reach a leaf after scanning $Q[i'..i-1]$, we continue the search with $Q[1..i'-1]$ on trie $\mathcal{R}(v')$, where $v'\in V$ is the node at block $k-1$ of $G(S)$ corresponding to the leaf we reached in the trie. If the search succeeds after reading $Q[1]$ we have found a path in $G(S)$ spelling $Q[1..j]$. We repeat the analogous procedure with $Q[j..m]$ starting from trie $\mathcal{F}(v)$.
That is, we can verify a candidate occurrence of $Q$ in $G(S)$ in $O(|Q|\log \sigma)$ time, as the search in the tries takes $O(\log \sigma)$ time per step. 

We are now ready to specify a theorem that reformulates Proposition~\ref{prop:blockrepeatfree} in detailed form.

\begin{theorem}
Let $G=(V,E)$ be a repeat-free founder graph with blocks $V^1, V^2, \ldots, V^b$ such that $V=V^1 \cup V^2 \cup \cdots \cup V^b$. We can preprocess an index structure for $G$ in $O((N+L|E|)\log \sigma)$ time, where $\{1,\ldots,\sigma\}$ is the alphabet for node labels, $L=\max_{v \in V} |\ell(v)|$, $N=\sum_{v \in V} |\ell(v)|$, and $\sigma\leq N$. Given a query string $Q[1..q] \in \{1,\ldots,\sigma\}^q$, we can use the index structure to find out if $Q$ occurs in $G$. This query takes $O(|Q| \log \sigma)$ time. 
\label{thm:blockrepeatfree}
\end{theorem}
\begin{proof}
To see that the approach stated above works correctly, we need to show that it suffices to verify exactly one arbitrary candidate occurrence identified by the Aho-Corasick automaton. Indeed, for contradiction, assume our verification fails in finding $Q$ starting from a candidate match $Q[i..j]=\ell(v)$, but there is another candidate match $Q[i'..j']=\ell(w)$, $v\neq w$, resulting in an occurrence of $Q$ in $G$. First, we can assume it holds that $[i..j]$ is not included in $[i'..j']$ and $[i'..j']$ is not included in $[i..j]$, since such nested cases would contradict the repeat-free property of $G$. Now, starting from the candidate match $Q[i'..j']$, we will find an occurrence of $Q[i..j]$ in $G$ when extending to the left or to the right. This occurrence cannot be $\ell(v)$, as we assumed the verification starting from $Q[i..j]=\ell(v)$ fails. That is, we found another occurrence of $\ell(v)$ in $G$, which is a contradiction with the repeat-free property. Hence, the verification starting from an arbitrary candidate match is sufficient.

With preprocessing time $O(N\log \sigma)$ we can build the Aho-Corasick automaton \cite{AC75}. The tries can be built in $O(\log \sigma)(\sum_{v \in V} (\sum_{(u,v) \in E} |\ell(u)| + \sum_{(v,w) \in E} |\ell(w)|))= O(|E|L\log \sigma)$ time. The search for a candidate match and the following verification take $O(|Q| \log \sigma)$ time.

We are left with the case of short queries not spanning a complete node label. To avoid the costly binary search in sorted $\mathcal{P}'$, we instead construct the unidirectional BWT index \cite{BCKM20} for the concatenation $C=\prod_{i \in \{1,2,\ldots,b\}}$ $\prod_{v \in V^i, (v,w) \in E} $ $\ell(v)\ell(w)0$.
Concatenation $C$ is thus a string of length $O(|E|L)$ from alphabet $\{\mathbf{0},1,2,\ldots, \sigma\}$. The unidirectional BWT index for $C$ can be constructed in $O(|C|)$ time, so that in $O(|Q|)$ time, one can find out if $Q$ occurs in $C$ \cite{BCKM20}. This query equals that of binary search in $\mathcal{P}'$.
\end{proof}

The above result can also be applied to \emph{degenerate strings} \cite{Alzetal20}. These are special case of elastic degenerate strings with equal length strings inside each block, and can thus be seen as fully connected block graphs.

\begin{corollary}
The results of Theorem \ref{thm:blockrepeatfree} hold for a \emph{repeat-free degenerate string} a.k.a. a fully connected repeat-free founder graph.
\label{thm:blockrepeatfreeGDstring}
\end{corollary}

Observe that $N<|C|\leq 2mn$, where $C$ is the concatenation in the proof above (whose length was bounded by $O(L|E|)$), and $m$ and $n$ are the number of rows and number of columns, respectively, in the  multiple sequence alignment from where the founder graph is induced. That is, the index construction algorithms of the above theorems can be seen to be take time almost linear in the (original) input size, namely, $O(mn\log \sigma)$ time. We study succinct variants of these indexes in Sect.~\ref{sect:compressedindexing}, and also improve the construction and query times to linear as side product.

\section{Construction of repeat-free founder graphs \label{sect:construction}}

Now that we know how to index repeat-free founder graphs, we turn our attention to the construction of such graphs from a given and \msa. For this purpose, we will adapt the dynamic programming segmentation algorithms for founders \cite{NCKM19,CKMN19}.

The idea is as follows. Let $S$ be a segmentation of $\msa[1..m,1..n]$. We say $S$ is \emph{valid} (or repeat-free) if it induces a repeat-free founder graph $G(S)=(V,E)$. A segment in $S$ is valid (or repeat-free) if $S$ is valid. We build such valid $S$ considering valid segmentation of prefixes of \msa from left to right, looking at shorter valid segmentations appended with a valid new segment.

\subsection{Characterization lemma}

Given a segmentation $S$ and founder graph $G(S)=(V,E)$ induced by $S$, we can ensure that it is valid by checking if, for all $v\in V$, $\ell(v)$ occurs in the rows of the \msa only in the interval of the block $V^i$, where $V^i$ is the block of $V$ such that $v \in V^i$.

\begin{lemma}[Characterization]
\label{lemma:charaterization}
Let $\mathcal{P}=[x_1..y_1],[x_2..y_2],\ldots,[x_b..y_b]$ be the partitioning corresponding to a segmentation $S$ inducing a block graph $G=(V,E)$. The segmentation $S$ is valid if and only if, for all blocks $V^i \subseteq V$, $1 \leq t \leq m$ and $j \neq x_i$, if $v \in V^i$ then $\msaij{t}{j\ltdots j+|\ell(v)|-1} \neq \ell(v)$.
\end{lemma}
\begin{proof}
To see that this is a necessary condition for the validity of $S$, notice that each row of the \msa can be read through $G$, so if $\ell(v)$ occurs elsewhere than inside the block, then these extra occurrences make $S$ invalid. To see that this is a sufficient condition for the validity of $S$, we observe the following:
\begin{enumerate}[a)]
    \item For all $(v,w)\in E$, $\ell(v)\ell(w)$ is a substring of some row of the input \msa.
    \item Let $(x,u),(u,y)\in E$ be two edges such that $U=\ell(x)\ell(u)\ell(y)$ is not a substring of any row of input \msa. Then any substring of $U$ either occurs in some row of the input \msa or it includes $\ell(u)$ as its substring.
    \item Thus, any substring of a path in $G$ either is a substring of some row of the input \msa, or it includes $\ell(u)$ of case b) as its substring.
    \item Let $\alpha$ be a substring of a path of $G$ that includes $\ell(u)$ as its substring. If $\ell(z)=\alpha$ for some $z\in V$, then $\ell(u)$ appears at least twice in the \msa. Substring $\alpha$ makes $S$ invalid only if $\ell(u)$ does.
\end{enumerate}
\end{proof}

\subsection{From characterization to a segmentation\label{sect:recurrence}}

Among the valid segmentations, we wish to select an \emph{optimal} segmentation under some goodness criteria. 

We consider three score functions for the valid segmentations, one maximizing the number of blocks, one minimizing the maximum height of a block, and one minimizing the maximum length of a block. The latter two have been studied earlier without the repeat-free constraint, and non-trivial linear time solutions have been found~\cite{NCKM19,CKMN19}, while the first score function makes sense only with this new constraint. 

Let $s(j')$ be the score of an optimal scoring segmentation $S^1,S^2,\ldots,S^b$ of prefix $\msaij{1..m}{1..j'}$ for a selected scoring scheme. Then  
\begin{equation}
s(j)=\bigoplus_{\begin{array}{c} j':0\leq j'< j,\\ \msaij{1..m}{j'+1..j} \text{ is }\\ 
\text{repeat-free segment} \end{array}} w(s(j'),j',j), 
\label{eq:valid-segmentation}
\end{equation}
gives the score of an optimal scoring repeat-free segmentation $S^1,S^2,\ldots,$ $S^b,S^{b+1}$ of $\msaij{1..m}{1..j}$,
where $\bigoplus$ is an operator depending on the scoring scheme and $w(x,j',j)$ is a function on the score $x$ of the segmentation of $S^1,S^2,\ldots,S^b$ and on the last block $S^{b+1}$ corresponding to $\msaij{1..m}{j'+1..j}$. To fix this recurrence so that $s(n)$ equals the maximum number of blocks over valid segmentations of $\msaij{1..m}{1..n}$, set $\bigoplus=\max$ and $w(x,j',j)=x+1$. For initialization, set $s(0)=0$. Moreover, when there is no valid segmentation for some $j$, set $s(j)=-\infty$. To fix this recurrence so that $s(n)$ equals the minimum of maximum heights of blocks over valid segmentations of $\msaij{1..m}{1..n}$, set $\bigoplus=\min$ and $w(x,j',j)=\max(x,|\{\msaij{i}{j'+1..j} \mid 1\leq i\leq m\}|)$. For initialization, set $s(0)=0$. Moreover, when there is no valid segmentation for some $j$, $s(j)=\infty$. Finally, to fix this recurrence so that $s(n)$ equals the minimum of maximum length of blocks over valid segmentations of $\msaij{1..m}{1..n}$, set $\bigoplus=\min$ and $w(x,j',j)=\max(x,j-j')$. For initialization, set $s(j)=0$. Moreover, when there is no valid segmentation for some $j$, set $s(j)=\infty$.

To derive efficient dynamic programming recurrences for these scoring functions, we separate the computation into the preprocessing phase and into the main computation. In the preprocessing phase, we compute values $v(j)$ and $f(j)$, $1\leq j\leq n$, defined as follows. Value $v(j)$ is the largest integer such that segment $\msaij{1..m}{v(j)+1..j}$ is valid. Value $f(j)$ is the smallest integer such that segment $\msaij{1..m}{j+1..f(j)}$ is valid. Note that $v(j)$ may not be defined for small $j$ and $f(j)$ may not be defined for large $j$ (short blocks may not be repeat-free).

Assuming values $v(j)$ have been preprocessed, we can simplify recurrence~(\ref{eq:valid-segmentation}) into
\begin{equation}
s(j)=\bigoplus_{j':0\leq j'\leq v(j)} w(s(j'),j',j),
\label{eq:valid-segmentation-vj}
\end{equation}
by observing that left-extensions of valid segments are also valid. We use this equation later for deriving a linear time solution for minimizing the maximum length of a block score.

With values $f(j)$ we can use an analogous observation that right-extensions of valid segments are also valid. This observation directly yields forward-propagation dynamic programming solutions for maximizing the number of blocks score and for minimizing the maximum length of a block score. These are given in Algorithms~\ref{algo:maxblocks}~and~\ref{algo:minmaxlength}. We leave it for future work to derive similar result for minimizing the maximum height of a block score.

\begin{algorithm}
\KwIn{Right-extensions $(j,f(j))$ sorted from smallest to largest order by second component: $(j_1,f(j_1)),(j_2,f(j_2)), \ldots , (j_{n-J},f(j_{n-J}))$, where $J$ is such that $f(j_{n-J+1}),f(j_{n-J+2}),\ldots, f(j_n)$ are not defined.}
\KwOut{Score of an optimal repeat-free segmentation maximizing the number of blocks.}
$x \gets 1$\; 
$\mathtt{maxblocks}(0) \gets 0$\; $\mathtt{maxblocks}(j) = -\infty$\; $\mathtt{maxscore} = -\infty$ for all $0<j\leq n$\;
\For{$j \gets 1$ to $n$}{
  \While{$j = f(j_{x})$}{
     $\mathtt{maxscore} \gets \max(\mathtt{maxscore},\mathtt{maxblocks}(j_{x}))$\;
     $x \gets x+1$\;
  }
  $\mathtt{maxblocks}(j) \gets \mathtt{maxscore}+1$\;
}
\Return{$\mathtt{maxblocks}(n)$}\;
\caption{\label{algo:maxblocks}
An $O(n)$ time algorithm for finding an optimal repeat-free segmentation maximizing the number of blocks.}
\end{algorithm}

\begin{algorithm}
\KwIn{Right-extensions $(j,f(j))$ sorted from smallest to largest order by second component: $(j_1,f(j_1)),(j_2,f(j_2)), \ldots , (j_{n-J},f(j_{n-J}))$, where $J$ is such that $f(j_{n-J+1}),f(j_{n-J+2}),\ldots, f(j_n)$ are not defined.}
\KwOut{Score of an optimal semi-repeat-free segmentation minimizing the maximum segment length.}
Initialize one-dimensional search trees $\mathcal{T}$ and $\mathcal{I}$ with keys $0,1,2,\ldots,2n$, with all keys associated with values $\infty$\;
$x \gets 1$\;
$\mathtt{minmaxlength}(0) \gets 0$\;
\For{$j \gets 1$ to $n$}{
  \While{$j = f(j_{x})$}{
     $\mathcal{T}.\mathtt{Upgrade}(j_{x}+\mathtt{minmaxlength}(j_{x}),\mathtt{minmaxlength}(j_{x}))$\;
     $\mathcal{I}.\mathtt{Upgrade}(j_{x}+\mathtt{minmaxlength}(j_{x}),-j_{x})$\;
     $x \gets x+1$\;
  }
  $\mathtt{minmaxlength}(j) \gets \min(\mathcal{T}.\mathtt{RangeMin}(j+1,\infty),\mathcal{I}.\mathtt{RangeMin}(-\infty,j)+j)$\;
}
\Return{$\mathtt{minmaxlength}(n)$}\;
\caption{\label{algo:minmaxlength}
An $O(n \log n)$ time algorithm for finding an optimal semi-repeat-free segmentation minimizing the maximum segment length. Minimization over an empty set is assumed to return $\infty$. Operation $\mathtt{Upgrade}(k,v)$ sets key $k$ to value $v$ if the previous value is larger. Operation $\mathtt{RangeMin}(a,b)$ returns the smallest value associated with keys in range $[a..b]$. Both operations can be supported in $O(\log n)$ time with standard balanced search trees.}
\end{algorithm}

\begin{theorem}
    After an $O(mn)$ time preprocessing, Algorithms~\ref{algo:maxblocks} and \ref{algo:minmaxlength} compute the scores $\mathtt{maxblocks}(n)=b$ and $\mathtt{minmaxlength}(n)=\max\limits_{i:1\leq i \leq b} L(S^i)$ of optimal repeat-free segmentations $S^1,S^2,\ldots,S^b$ of gapless $\msaij{1..m}{1..n}$ in $O(n)$ and $O(n \log \log n)$ time, respectively. The produced segmentations induce repeat-free founder graphs from a gapless \msa.
\end{theorem}
\begin{proof}
The $O(mn)$ preprocessing algorithm is provided in Theorem~\ref{th:vj}. Let us then consider the running time of the main algorithms. In both algorithms, the sorted input can be produced in $O(n)$ time by counting sort from the output of the preprocessing algorithm. The first algorithm takes clearly linear time. The second algorithm takes clearly $O(n \log n)$ time, but this can be improved: Since the queries are semi-open intervals with keys in range $[1 \ldots 2n]$, these balanced search trees can be replaced by van Emde Boas trees to obtain $O(n\log \log n)$ time computation of all values~\cite{GBT84}. 

Correctness of Algorithms~\ref{algo:maxblocks} and \ref{algo:minmaxlength} follow from the fact that when computing the score at column $j$, all earlier segmentations that are safe to be extended with a new segment ending at $j$ are considered. We formalize this argument for Algorithm~\ref{algo:minmaxlength}, as the proof for the other is analogous and easier. Assume by induction that $\mathtt{minmaxlength}(j')$ is the score $\max\limits_{i:1\leq i \leq b} L(S^i)$ of an optimal semi-repeat-free segmentation $S^1,S^2,\ldots,S^b$ of $\msaij{1..m}{1..j'}$, for $j'<j$. Each $\mathtt{minmaxlength}(j')$ is added to the data structures when the corresponding segmentation can be considered to be appended with segment $S^{b+1}$ corresponding to $\msaij{1..m}{j'+1..j}$, for $j\geq f(j')$, so that the result is a semi-repeat-free segmentation. The minimum values from the data structures equal the definition of segmentation score $\max\limits_{i:1\leq i \leq b+1} L(S^i)$: To see this, we have two cases to consider: a) for $j'$ such that $\mathtt{minmaxlength}(j')>j-j'$ the score of the segmentation ending at $j'$ extended with $[j'+1..j]$ is $\mathtt{minmaxlength}(j')$, and b) for $j'$ such that $\mathtt{minmaxlength}(j')\leq j-j'$ the score of the segmentation ending at $j'$ extended with $[j'+1..j]$ is $j-j'$. The query intervals guarantee that the minima is returned, with the latter adjusted by $+j$ so that it gives $j-j’$ for minimum $-j$ in tree tree $\mathcal{I}$, corresponding to the cases a) and b). Initialization guarantees that score of the first segment is correctly computed. Traceback from $\mathtt{minmaxlength}(n)$ gives an optimal semi-repeat-free segmentation.
\end{proof}

\subsection{Preprocessing\label{sect:preprocessing}}

We can do the preprocessing for values $v(j)$ and $f(j)$ in $O(mn)$ time. The idea is to build a BWT index on the \msa rows, and then search all rows backward from right to left in parallel (with everything reversed for the latter values). Once we reach a column $j'$ where all suffixes have altogether exactly $m$ occurrences (their union of BWT intervals is of size $m$), then $\msaij{1..m}{j'..n}$ is a valid segment. Then we can drop the last column (do right-contract on all rows) and continue left-extensions until finding the largest $j'$ such that $\msaij{1..m}{j'..n-1}$ is a valid segment. Continuing this way, we can find for each column $j$ the value $v(j)=j'-1$. The bottleneck of the approach is the computation of the size of the union of intervals, but we can avoid a trivial computation by exploiting the repeat-free property and the order in which these intervals are computed.

\begin{theorem}\label{th:vj}
Given a multiple sequence alignment \msamn with each $\msaij{i}{j} \in [1..\sigma]$, values $v(j)$ and $f(j)$ for each $1 \leq j \leq n$ can be computed in randomized $O(mn)$ time or deterministic $O(mn\log \sigma)$ time. Here value $v(j)$ is the largest integer such that segment $\msaij{1..m}{v(j)+1..j}$ is valid, and value $f(j)$ is the smallest integer such that segment $\msaij{1..m}{j+1..f(j)}$ is valid.
\end{theorem}
\begin{proof}
We consider values $v(j)$ as the other case is symmetric. Let us build the bidirectional BWT index \cite{BCKM20} of \msa rows concatenated into one long string with some separator symbols added between rows. We will run several phases in synchronization over this BWT index, but we explain them first as if they would be run independently.

Phase 0 searches in parallel all rows from right to left advancing each by one position at a time. Let $k$ be the number of parallel of steps done so far. We can maintain a bitvector $M$ that at the $k$-th step stores $M[i]=1$ iff $BWT[i]$ is the $k$-th last symbol of some row.  

Phase 1 uses the \emph{variable length sliding window} approach of Belazzougui and Cunial \cite{BC19} to compute values $v(j)$. Let the first row of \msa be $T[1\ltdots n]$. Search $T[1\ltdots n]$ backwards in the fully-functional bidirectional BWT index \cite{BC19}. Stop the search at $T[j'+1\ltdots n]$ such that the corresponding BWT interval $[i'\ltdots i]$ contains only suffixes originating from column $j'+1$ of the \msa, that is, spelling $\msaij{a}{j'+1\ltdots n}$ in the concatenation, for some rows $a$. Set $v^b(n)=j'$ for row $b=1$. Contract $T[n]$ from the search string and modify BWT interval accordingly \cite{BC19}. Continue the search (decreasing $j'$ by one each step) to find $T[j'+1\ltdots n-1]$ s.t. again the corresponding BWT interval $[i'\ltdots i]$ contains only suffixes originating from column $j'+1$. Update $v^b(n-1)=j'$ for row $b=1$. Continue like this throughout $T$. Repeat the process for all remaining rows $b\in [2..m]$, to compute $v^2(j),v^3(j),\ldots,v^m(j)$ for all $j$. 
Set $v(j)=\min_i v^i(j)$ for all $j$. 

Let us call the instances of Phase 1 run on the rest of the rows as Phases $2, 3, \ldots, m$. 

Let the current BWT interval in Phases $1$ to $m$ be $[j'+1\ltdots j]$. The problematic part in them is checking if the corresponding \emph{active} BWT intervals $[i'_a\ltdots i_a]$ for Phases $a \in \{1,2,\ldots, m\}$ contain only suffixes originating from column $j'+1$. To solve this, we run Phase 0 as well as Phases $1$ to $m$ in synchronization so that we are at the $k$-th step in Phase 0 when we are processing interval $[j'+1\ltdots j]$ in the rest of the Phases, for $k=n-j'$. In addition, we maintain bitvectors $B$ and $E$ such that $B[i'_a]=1$ and $E[i_a]=1$ for $a \in \{1,2,\ldots, m\}$. For each $M[i]$ that we set to 1 at step $k$ with $B[i]=0$ and $E[i]=0$, we check if $M[i-1]=1$ and $M[i+1]=1$. If and only if this check fails on any $i$, there is a suffix starting outside column $j'+1$. This follows from the fact that each suffix starting at column $j'+1$ must be contained in exactly one of the distinct intervals of the set $I=\{[i'_a\ltdots i_a]\}_{a\in \{1,2 \ldots m\}}$. This is because $I$ cannot contain nested interval pairs as all strings in segment $[j'+1\ltdots j]$ of the \msa are of equal length, and thus their BWT intervals cannot overlap except if the intervals are exactly the same.  

Finally, the running time of the algorithm is $O(mn)$ or $O(mn\log \sigma)$, using  Lemma~\ref{lemma:randomized_index_construction} or Lemma~\ref{lemma:index_construction}, respectively, and using the fact that the bitvectors are manipulated locally only on indexes that are maintained as variables during the execution.
\end{proof}

\subsection{Faster algorithm for minimizing the maximum block length\label{sect:fastermainalgorithm}}

Recall Eq.~(\ref{eq:valid-segmentation-vj}). Let us consider the score $w(x,j',j)=\max(s(j'),j-j')$ with $\bigoplus=\min$, that is, minimizing the maximum block length over valid segmentations. Algorithm~\ref{algo:minmaxlength} solved this problem in near-linear time, but now we improve this to linear using values $v(j)$ instead of $f(j)$. 
The basic observation is that $v(J)\leq v(J+1) \leq \cdots \leq v(n)$, for some $J>0$, and hence the range where the minimum is taken grows as $j$ grows.

Cazaux et al. \cite{CKMN19} considered a similar recurrence and gave a linear time solution for it. In what follows we modify that technique to work with valid ranges.

For $j$ between $1$ and $n$, we define
\[x(j) = \max \argmin_{j' \in [1\ltdots v(j)]} \max(j-j',s(j'))\]

\begin{lemma}
    For any $j\in [1\ltdots n-1]$, we have $x(j) \leq x(j+1)$.
\end{lemma}

\begin{proof}
By the definition of $x(.)$, for any $j\in [1\ltdots n]$, we have for $j' \in [1\ltdots x(j)-1]$, $\max(j-j',s(j')) \geq \max(j-x(j),s(x(j)))$ and for $j' \in [x(j)+1 \ltdots v(j)]$, $\max(j-j',s(j')) > \max(j-x(j),s(x(j)))$.

We assume that there exists $j \in [1\ltdots n-1]$, such that $x(j+1) < x(j)$. In this case, $x(j+1) \in [1\ltdots x(j)-1]$ and we have $\max(j-x(j+1),s(x(j+1))) \geq \max(j-x(j),s(x(j)))$. As $v(j+1) \geq v(j)$, $x(j) \in [x(j+1)+1\ltdots v(j+1)]$ and thus $\max(j+1-x(j+1),s(x(j+1))) < \max(j+1-x(j),s(x(j)))$. As $x(j+1) < x(j)$, we have $j-x(j+1) > j-x(j)$.
To simplify the proof, we take $A = j-x(j+1)$, $B = s(x(j+1))$, $C= j - x(j)$ and $D = s(x(j))$. Hence, we have $\max(A,B) \geq \max(C,D)$, $\max(A+1,B) < \max(C+1,D)$ and $A > C$. Now we are going to prove that this system admits no solution.
\begin{itemize}
    \item Case where $A = \max(A,B)$ and $C = \max(C,D)$. As $A> C$, we have $A+1 > C+1$ and thus $\max(A+1,B) > \max(C+1,D)$ which is impossible because $\max(A+1,B) < \max(C+1,D)$.
    \item Case where $B = \max(A,B)$ and $C = \max(C,D)$. We can assume that $B>A$ (in the other case, we take $A = \max(A,B)$) and as $A > C$, we have $B > C+1$ and thus $\max(A+1,B) > \max(C+1,D)$ which is impossible because $\max(A+1,B) < \max(C+1,D)$.
    \item Case where $A = \max(A,B)$ and $D = \max(C,D)$. We have $A > D$ and $A > C$, thus $\max(A+1,B) > \max(C+1,D)$ which is impossible because $\max(A+1,B) < \max(C+1,D)$.
    \item Case where $B = \max(A,B)$ and $D = \max(C,D)$. We have $B \geq D$ and $A > C$, thus $\max(A+1,B) \geq \max(C+1,D)$ which is impossible because $\max(A+1,B) < \max(C+1,D)$.
\end{itemize}
\end{proof}

\begin{lemma}\label{le:bound:s}
    By initialising $s(1)$ to a threshold $K$, for any $j\in [1\ltdots n]$, we have $s(j) \leq \max(j,K)$.
\end{lemma}

\begin{proof}
    We are going to show this by induction. The base case is obvious because $s(1) = K \leq \max(1,K)$.
    As $s(j)=\min_{j':1\leq j'\leq v(j)} \max(j-j',s(j'))$, by using induction, $s(j) \leq \min_{j':1\leq j'\leq v(j)} \max(j,K) \leq \max(j,K)$
\end{proof}

Thanks to Lemma~\ref{le:bound:s}, by taking the threshold $K = n+1$, the values $s(j)$ are in $O(n)$ for all $j$ in $[1\ltdots n-1]$.

\begin{lemma}\label{le:j:star}
    Given $j^{\star} \in [x(j-1)+1\ltdots v(j)]$, we can compute in constant time if
    \[
    j^{\star} = \max \argmin_{j' \in [j^{\star} \ltdots v(j)]} \max(j-j',s(j')).
    \]
\end{lemma}

\begin{proof}
    We need just to compare $k = \max(j-j^{\star},s(j^{\star}))$ and $s(j^{\diamond})$ where $j^{\diamond}$ is in \\$\argmin_{j' \in [j^{\star}+1\ltdots v(j)]} s(j')$. If $k$ is smaller than $s(j^{\diamond})$, $k$ is smaller than all the $s(j')$ with $j' \in [j^{\star}+1\ltdots  v(j)]$ and thus for all $\max(j-j',s(j'))$. Hence we have \\$j^{\star}= \max \argmin_{j' \in [j^{\star} \ltdots v(j)]} \max(j-j',s(j'))$.

    \begin{sloppypar}
    Otherwise, $s(j^{\diamond}) \geq k$ and as $k \geq j-j^{\star}$,  $\max(j-j^{\diamond},s(j^{\diamond})) \geq k$. In this case $j^{\star} \neq \max \argmin_{j' \in [j^{\star}\ltdots  v(j)]} \max(j-j',s(j'))$.
    By using the constant time semi-dynamic range maximum query by Cazaux et al. \cite{CKMN19} on the array $s(.)$, we can obtain in constant time $j^{\diamond}$ and thus check the equality in constant time.
    \end{sloppypar}
\end{proof}

\begin{theorem}
    The values $s(j)$ of Eq.~(\ref{eq:valid-segmentation-vj}) with $w((s(j'),j',j)=\max(s(j'),j-j')$ and $\bigoplus=\min$ for all $j \in [1 \ltdots n]$, can be computed in $O(n)$ time after a randomized $O(mn)$ time or deterministic $O(mn\log \sigma)$ time preprocessing on a gapless multiple sequence alignment. The optimal segmentation defined by Eq.~(\ref{eq:valid-segmentation-vj}) yields a repeat-free founder graph.
    \label{thm:linearsegmentation}
\end{theorem}

\begin{proof}
    We begin by preprocessing all the values of $v(j)$ in $O(mn)$ randomized or $O(mn\log \sigma)$ deterministic time (Theorem~\ref{th:vj}).
    The idea is to compute all the values $s(j)$ by increasing order of $j$ and by using the values $x(j)$. For each $j \in [1\ltdots n]$, we check all the $j'$ from $x(j-1)$ to $v(j)$ with the equality of Lemma~\ref{le:j:star} until one is true and thus corresponds to $x(j)$.  Finally, we add $s(j) = \max(j-x(j),s(x(j)))$ to the constant time semi-dynamic range maximum query and continue with $j+1$.
\end{proof}

\section{Succinct index for repeat-free founder graphs \label{sect:compressedindexing}}

Recall the indexing solutions of Sect.~\ref{sect:repeat-freeness} and the definitions from Sect.~\ref{sect:definitions}. 

We now show that explicit tries and Aho-Corasick automaton can be replaced by some auxiliary data structures associated with the Burrows-Wheeler transformation of the concatenation $C=\prod_{i \in \{1,2,\ldots,b\}}$ $\prod_{v \in V^i, (v,w) \in E}$ $\ell(v)\ell(w)\mathbf{0}$. 

Consider interval $\mathtt{SA}[i\ltdots k]$ in the suffix array of $C$ corresponding to suffixes having $\ell(v)$ as prefix for some $v \in V$. From the repeat-free property it follows that this interval can be split into two subintervals, $\mathtt{SA}[i\ltdots j]$ and $\mathtt{SA}[j+1\ltdots k]$, such that suffixes in $\mathtt{SA}[i\ltdots j]$ start with $\ell(v)\mathbf{0}$ and suffixes in $\mathtt{SA}[j+1\ltdots k]$ start with $\ell(v)\ell(w)$, where $(v,w) \in E$. Moreover, BWT$[i\ltdots j]$ equals multiset $\{\ell(u)[|\ell(u)|-1] \mid (u,v) \in E\}$ \emph{sorted in lexicographic order}. This follows by considering the lexicographic order of suffixes $\ell(u)[|\ell(u)|-1]\ell(v)\mathbf{0}\ldots$ for $(u,v)\in E$. That is, BWT$[i\ltdots j]$ (interpreted as a set) represents the children of the root of the trie $\mathcal{R}(v)$ considered in Sect.~\ref{sect:repeat-freeness}.

We are now ready to present the search algorithm that uses only the BWT of $C$ and some small auxiliary data structures. We associate two bitvectors $B$ and $E$ to the BWT of $C$ as follows. We set $B[i]=1$ and $E[k]=1$ iff $\mathtt{SA}[i\ltdots k]$ is maximal interval with all suffixes starting with $\ell(v)$ for some $v\in V$.

Consider the backward search with query $Q[1\ltdots q]$. Let $\mathtt{SA}[j'\ltdots k']$ be the interval after matching the shortest suffix $Q[q'\ltdots q]$ such that BWT$[j']=\mathbf{0}$. Let $i=\mathtt{select}(B,\mathtt{rank}(B,j'))$ and $k=\mathtt{select}(E,\mathtt{rank}(B,j'))$.
If $i\leq j'$ and $k'\leq k$, index $j'$ lies inside an interval $\mathtt{SA}[i\ltdots k]$ where all suffixes start with $\ell(v)$ for some $v$. We modify the range into $\mathtt{SA}[i\ltdots k]$, and continue with the backward step on $Q[q'-1]$. We check the same condition in each step and expand the interval if the condition is met. Let us call this procedure \emph{expanded backward search}.

We can now strictly improve Theorem~\ref{thm:blockrepeatfree}~and~Corollary~\ref{thm:blockrepeatfreeGDstring} as follows.

\begin{theorem}
Let $G=(V,E)$ be a repeat-free founder graph (or a repeat-free degenerate string) with blocks $V^1, V^2, \ldots, V^b$ such that $V=V^1 \cup V^2 \cup \cdots \cup V^b$. We can preprocess an index structure for $G$ occupying $O(L|E|\log\sigma)$ bits in $O(L|E|)$ time, where $\{1,\ldots,\sigma\}$ is the alphabet for node labels and $L=\max_{v \in V} \ell(v)$. Given a query string $Q[1\ltdots q] \in \{1,\ldots,\sigma\}^q$, we can use expanded backward search with the index structure to find out if $Q$ occurs in $G$. This query takes $O(|Q|)$ time. 
\label{thm:blockrepeatfreeBWTindex}
\end{theorem}
\begin{proof}
As we  expand the search interval in BWT, it is evident that we still find all occurrences for short patterns that span at most two nodes, like in the proof of Theorem~\ref{thm:blockrepeatfree}. We need to show that a) the expansions do not yield spurious occurrences for such short patterns and b) the expansions yield exactly the occurrences for long patterns that we earlier found with the Aho-Corasick and tries approach. 
In case b), notice that after an expansion step we are indeed in an interval $\mathtt{SA}[i\ltdots k]$ where all suffixes match $\ell(v)$ and thus corresponds to a node $v \in V$. The suffix of the query processed before reaching interval $\mathtt{SA}[i\ltdots k]$ must be at least of length $|\ell(v)|$. That is, to mimic Aho-Corasick approach, we should continue with the trie $\mathcal{R}(v)$. This is identical to taking a backward step from BWT$[i\ltdots k]$, and continuing therein to follow the rest of this implicit trie. 

To conclude case b), we still need to show that we reach all the same nodes as when using Aho-Corasick, and that the search to other direction with $\mathcal{L}(v)$ can be avoided. These follow from case a), as we see.

In case a), before doing the first expansion, the search is identical to the original algorithm in the proof of Theorem~\ref{thm:blockrepeatfree}. After the expansion, all matches to be found are those of case b). That is, no spurious matches are reported. Finally, no search interval can include two distinct node labels, so the search reaches the only relevant node label, where the Aho-Corasick and trie search simulation takes place. We reach all such nodes that can yield a full match for the query, as the proof of Theorem~\ref{thm:blockrepeatfree} shows that it is sufficient to follow an arbitrary candidate match.

As we only need a standard backward step, we can use a unidirectional BWT index constructable in deterministic $O(L|E|)$ time supporting a backward step in constant time \cite{BCKM20}.
\end{proof}

\section{Conditional hardness of indexing \efgs\label{sect:hardness}}

\begin{figure*}[t]
    \centering
    \includegraphics[scale=0.9]{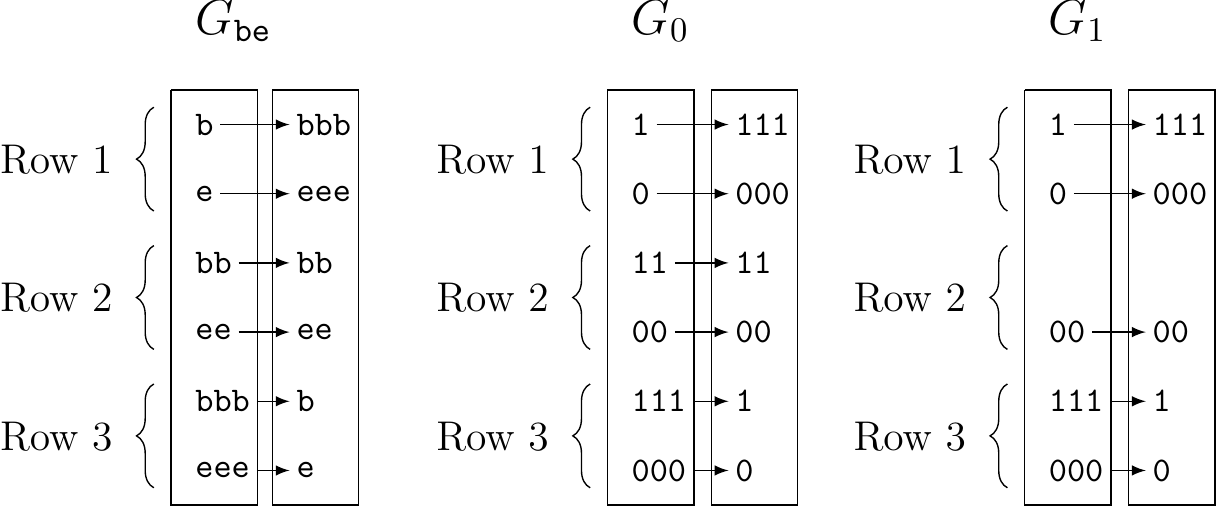}
    \caption{Gadgets $G_{\BB\E}$, $G_{\zero}$ and $G_\one$. Each gadget is organized into three rows, each row encoding a different partitioning of the strings $\BB\BB\BB\BB$, $\E\E\E\E$, $\zero\zero\zero\zero$, $\one\one\one\one$. This ensures that, when combining these gadgets in \Cref{fig:G}, edges can be controlled to go within the same row, or to the row below.}
    \label{fig:G_be01}
\end{figure*}
We now turn our attention to general \msas and elastic founder graphs induced from their segmentation. With non-elastic founder graphs, we have seen that the repeat-free property makes them indexable, but for now we have no proof that such property is necessary. For elastic founder graphs, we are able to derive such a conditional lower bound.

Namely, we show a reduction from \emph{Orthogonal Vectors} (\ov) to the problem of matching a query string in an \efg, continuing the line of research conducted on many related (degenerate) string problems \cite{CR16,EGMT19,Alzetal20,GibneySPIRE2020}. The \ov problem is to find out if there exist $x \in X$ and $y \in Y$ such that $x \cdot y = 0$, given two set $X$ and $Y$ of $n$ binary vectors each. We construct string $Q$ using $X$ and graph $G$ using $Y$. Then, we show that $Q$ has a match in $G$ if and only if $X$ and $Y$ form a ``yes''-instance of \ov. We condition our results on the following \ov hypothesis, which is implied by the \emph{Strong Exponential Time Hypothesis}~\cite{IP01}.

\begin{definition}[Orthogonal Vectors Hypothesis (\ovh)~\cite{Williams05}]
\label{def:OVH}
Let $X, Y$ be the two sets of an \ov instance, each containing $n$ binary vectors of length $d$.\footnote{In this section, keeping in line with the usual notation in the \ov problem, we use $n$ to denote the size of $X$ and $Y$, instead of the number of columns of the \msa.} For any constant $\epsilon>0$, no algorithm can solve \ov in time $O(\text{poly}(d)n^{2-\epsilon})$.
\end{definition}

\subsection{Query string} 
We build string $Q$ by combining string gadgets $Q_1, \ldots, Q_n$, one for each vector in $X$, plus some additional characters. To build string $Q_i$, first we place four \BB characters, then we scan vector $x_i \in X$ from left to right. For each entry of $x_i$, we place sub-string $Q_{i,h}$ consisting of four \zero characters if $x_i[h]=0$, or four \one characters if $x_i[h]=1$. Finally, we place four \E characters. For example, vector $x_i = 101$ results into string
\[Q_i = \BB\BB\BB\BB \, Q_{i,1} \, Q_{i,2} \, Q_{i,3} \, \E\E\E\E
\text{, where }\;
Q_{i,1} = \one\one\one\one,
\]
\[\quad Q_{i,2} = \zero\zero\zero\zero, \quad Q_{i,3} = \one\one\one\one.\]
Full string $Q$ is then the concatenation $Q = \BB\BB\BB\BB Q_1 Q_2 \ldots Q_n \E\E\E\E$.
The reason behind these specific quantities will be clear when discussing the structure of the graph.

\subsection{Elastic founder graph}
We build graph $G$ combining together three different sub-graphs: $G_L$, $G_M$, $G_R$ (for \emph{left}, \emph{middle} and \emph{right}).
Our final goal is to build a graph structured in three logical ``rows''. We denote the three rows of $G_M$ as $G_{M1}$, $G_{M2}$, $G_{M3}$, respectively. The first and the third rows of $G$, along with subgraphs $G_L$ and $G_R$ (introduced to allow slack), can match any vector. The second row matches only sub-patterns encoding vectors that are orthogonal to the vectors of set $Y$. The key is to structure the graph such that the pattern is forced to utilize the second row to obtain a full match. 
We present the full structure of the graph in \Cref{fig:G}, which shows the graph built on top of vector set $\{100,\,011,\,010\}$. In particular, $G_M$ consists of $n$ gadgets $G_M^j$, one for each vector $y_j \in Y$. The key elements of these sub-graphs are gadgets $G_{\BB\E}$, $G_{\zero}$ and $G_\one$ (see \Cref{fig:G_be01}), which allow to stack together multiple instances of strings $\BB^4$, $\E^4$, $\one^4$, $\zero^4$. The overall structure mimics the one in~\cite{EGMT19}, except for the new idea from \Cref{fig:G_be01}. 

\begin{figure*}[t]
    \begin{subfigure}[b]{\textwidth}\centering
    \includegraphics[width=0.9\textwidth]{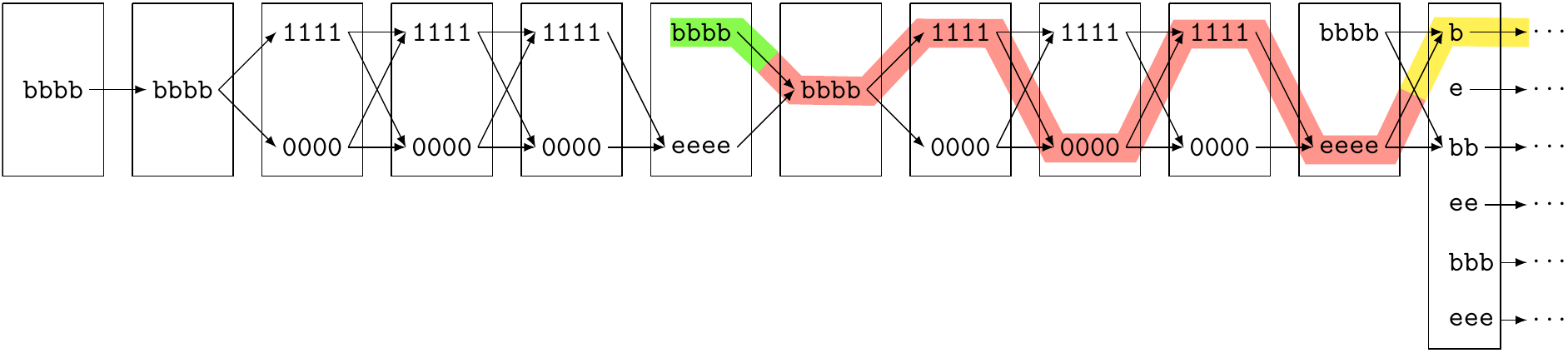}
    \caption{Sub-graph $G_L$. The last segment belongs to sub-graph $G_M$ and shows the connection.}
    \label{fig:G_L}
    \end{subfigure}
    \smallskip
    \begin{subfigure}[b]{\textwidth}
    \includegraphics[width=\textwidth]{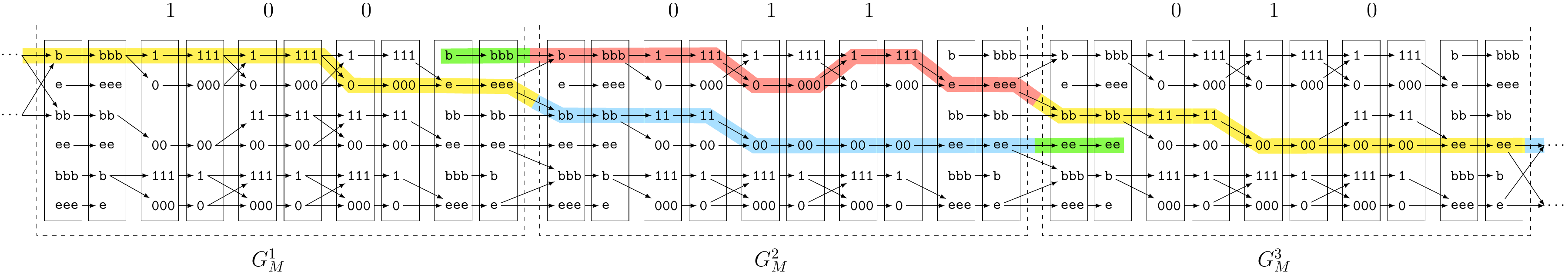}
    \caption{Sub-graph $G_M$ for vectors $y_1 = 100$, $y_2 = 011$ and  $y_3 = 010$. The dashed rectangles highlight the single $G_M^j$ gadgets.}
    \label{fig:G_M}
    \end{subfigure}
    \smallskip
    \begin{subfigure}[b]{\textwidth}
    \includegraphics[width=0.9\textwidth]{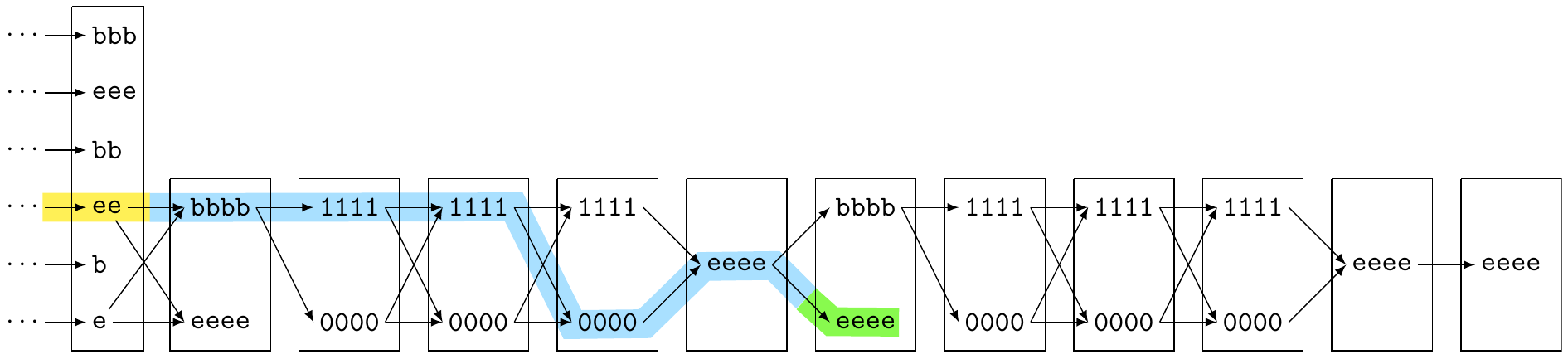}
    \caption{Sub-graph $G_R$. The first segment belongs to sub-graph $G_M$ and shows the connection.}
    \label{fig:G_R}
    \end{subfigure}
    \caption{An example of graph $G$. To visualize the entire graph, watch the three sub-figures from top to bottom and from left to right. 
    We also show two example occurrences of a query string $Q$ constructed from $x_1 = 101$, $x_2=110$, $x_3 = 100$ (left-most), and from $x_1 = 101$, $x_2 = 100$, $x_3 = 110$ (right-most), respectively. We highlight each $Q_i$ with a different color. Any such occurrence must pass through the middle row of $G_M$.
    \label{fig:G}}
\end{figure*}

\subsubsection{Detailed structure of the graph.}
\textbf{Sub-graph $G_L$} (Figure~\ref{fig:G_L}) consists of a starting segment with a single node labeled $\BB^4$, followed by $n-1$ sub-graphs $G_{L}^1, \ldots, G_{L}^{n-1}$, in this order. Each $G_{L}^i$ has $d+2$ segments, and is obtained as follows. First, we place a segment containing only one node with label $\BB^4$, then we place $d$ other segments, each one containing two nodes with labels $\one^4$ and $\zero^4$. Finally, we place a segment containing two nodes with labels $\BB^4$ and $\E^4$. 

The nodes in each segment are connected to all nodes in the next segment, with the exception of the last segment of each $G_{L}^i$: in this case, the node with label $\one^4$ and the one with label $\zero^4$ are connected only to the $\E^4$-node of the next (and last) segment of such $G_{L}^i$. 

\textbf{Sub-graph $G_R$} (Figure~\ref{fig:G_R}) is similar to sub-graph $G_L$, and it consists in $n-1$ parts $G_{R}^1, \ldots, G_{R}^{n-1}$, followed by a segment with a single node labeled $\E^4$. Part $G_{R}^i$ has $d+2$ segments, and is constructed almost identically to $G_{L}^i$. The differences are that, in the first segment of $G_{R}^i$, we place two nodes labeled $\BB^4$ and $\E^4$, while in the last segment we place only one node, which we label $\E^4$.

As in $G_L$, the nodes in each segment are connected to all nodes in the next segment, with the exception of the first segment of each $G_{R}^i$: in this case, the node labeled $\E^4$ has no outgoing edge.

\textbf{Sub-graph $G_M$} (Figure~\ref{fig:G_M}) implements the main logic of the reduction, and it uses three building blocks, $G_{\BB\E}$, $G_{\zero}$ and $G_\one$, which are organized in three rows, as shown in Figure~\ref{fig:G_be01}.

Sub-graph $G_M$ has $n$ parts, $G_M^1, \ldots, G_M^n$, one for each of the vectors $y_1, \ldots, y_n$ in set $Y$. Each $G_M^j$ is constructed, from left to right, as follows. First, we place a $G_{\BB\E}$ gadget. Then, we scan vector $y_j$ from left to right and, for each position $h \in \{1,\dots,d\}$, we place a $G_\zero$ gadget if the $h$-th entry is $y_j[h]=\zero$, or a $G_\one$ if $y_j[h]=\one$. Finally, we place another $G_{\BB\E}$ gadget. 

For the edges, we first consider each gadget $G_M^j$ separately. Let $G_h$ and $G_{h+1}$, be the gadgets encoding $y_j[h]$ and $y_j[h+1]$, respectively. We fully connect the nodes of $G_h$ to the nodes of $G_{h+1}$ row by row, respecting the structure of the segments. Then we connect, row by row, the $\BB$-nodes of the left $G_{\BB\E}$ to the leftmost $G_h$, which encodes $y_j[1]$, and the nodes of the rightmost $G_h$, which encodes $y_j[d]$, to the $\E$-nodes of the right $G_{\BB\E}$, again row by row. We repeat the same placement of the edges for every vector $G_h$, $G_{h+1}$, $1\leq h \leq d-1$; this construction is shown in Figure~\ref{fig:G_M}.

To conclude the construction of $G_M$, we need to connect all the $G_M^j$ gadgets together. Consider the right $G_{\BB\E}$ of gadget $G_M^j$, and the left $G_{\BB\E}$ of gadget $G_M^{j+1}$. The edges connecting these two gadgets are depicted in Figure~\ref{fig:G_M}, which shows that following a path we can either remain in the same row or move to the row below, but we cannot move to the row above. Moreover, sub-pattern $\BB^8$ can be matched only in the first and second row, while sub-pattern $\E^8$ only in the second and third rows.

In proving the correctness of the reduction, we will use $G_{M1}$, $G_{M2}$ and $G_{M3}$ to refer to the sub-graphs of $G_M$ consisting of only the nodes and edges of the first, second and third row, respectively. Formally, for $t \in \{1,2,3\}$, $V_{Mt} \subset V$ $V_{Mt} \subset V$ is the set of nodes placed in the $t$-th row of each $G_{\BB\E}$, $G_0$ or $G_1$ gadget belonging to sub-graph $G_M$, and \mbox{$E_{Mt} =\{(v,w) \in E \,|\, v,w \in V_{Mt}\}$}. Thus, \mbox{$G_{Mt} = (V_{Mt}, E_{Mt})$}. We will use the notation $G_{M2}^j$ to refer to the nodes belonging to both $G_M^j$ and $G_{M2}$, excluding the ones in $G_{M1}$ and $G_{M3}$, and the edges connecting them.

\textbf{Final graph $G$} is obtained by combining sub-graphs $G_L$, $G_M$ and $G_R$. To this end, we connect the nodes in the last segment of $G_L$ with the $\BB$-nodes in the first and second row of the left $G_{\BB\E}$ gadget of $G_M^1$. Finally, we connect the $\E$-nodes in the second and third row of the right $G_{\BB\E}$ gadget of $G_M^n$ with both the $\BB^4$-node and $\E^4$-node in the first segment of $G_R$. Figures~\ref{fig:G_L},~\ref{fig:G_M}~and~\ref{fig:G_R} can be visualized together, in this order, as one big picture of final graph $G$. In Figures~\ref{fig:G_L}~and~\ref{fig:G_R} we also included the adjacent segment of $G_M$ to show the connection.

\subsection{\ovh conditional hardness}
The proof of correctness is similar to the one in~\cite{EGMT19}, but with adaptations to the elastic founder graph. We prove three lemmas concerning $G_{M2}$, which are key for the correctness. The first lemma is a straightforward consequence of the structure of $G_{M2}$ and the fact that it is directed.

\begin{restatable}{lemma}{lemmasamej}
\label{lemma:samej}
If string $Q_i$ has a match in $G_{M2}$, then the path matching $Q_i$ is fully contained in $G_{M2}^j$, for some $1\leq j \leq n$. Moreover, each $Q_{i,h}$ sub-string matches a path of two nodes which belong to the $G_0$ or $G_1$ gadget encoding $y_j[h]$.
\end{restatable}

\begin{proof}
The claim follows by construction, since $Q_i$ starts with $\BB^4\zero^4$ or $\BB^4\one^4$ (which are found only at the beginning of a $G_{M2}^j$ gadget), and ends with $\zero^4\E^4$ or $\one^4\E^4$ (which are found only at the end of a $G_{M2}^j$).
\end{proof}

\begin{restatable}{lemma}{lemmamatchingGW}
\label{lemma:matchingGW}
String $Q_i$ has a match in $G_{M2}$ if and only if there exists $y_j \in Y$ such that $x_i \cdot y_j = 0$.
\end{restatable}

\begin{proof}
Recall that, by construction, the $h$-th $G_*$, $* \in \{0,1\}$ gadget in $G_{M2}^j$ if a $G_0$ gadget if and only if $y_j[h] = 0$, while it is a $G_1$ gadget if and only if $y_j[h] = 1$. We handle the two implications of the statement individually.

($\Rightarrow$) By Lemma~\ref{lemma:samej}, we can focus on the $d$ distinct and consecutive nodes of $G_{M2}^j$ that match $Q_i$. In particular we know that each sub-string $Q_{i,h}$ matches in the second row of either the $h$-th gadget $G_0$ or the $h$-th gadget $G_1$. Consider vectors $x_i \in X$ and $y_j \in Y$. If $Q_{i,h} = \one^4$ has a match in $G_{M2}^j$ it means that its $h$-th gadget is a $G_0$, and hence $y_j[h]=0$, implying $x_i[h] \cdot y_j[h] = 0$. If $Q_{i,h} = \zero$, by construction we know that $x_i[h]=0$, and $Q_{i,h}$ can have a match in $G_{M2}^j$ no matter whether the $h$-th gadget is a $G_0$ or a $G_1$. Thus, it clearly holds that $x_i[h] \cdot y_j[h] = 0$. At this point, we can conclude that $x_i[h] \cdot y_j[h] = 0$ for every $1 \leq h \leq d$, thus $x_i \cdot y_j = 0$.

($\Leftarrow$) Consider vectors $x_i \in X$ and $y_j \in Y$ that are such that $x_i \cdot y_j = 0$. For $h = 1, 2, \ldots, d$, if $y_j[h] = 0$ then the $h$-th gadget of $G_{M2}^j$ is a $G_0$ gadget, and $Q_{i,h}$ can surely match it. If $y_j[h] = 1$ it must hold that $x_i[h] = 0$, since $x_i \cdot y_j = 0$. Thus $Q_{i,h} = \zero^4$, and it can have a match in the $h$-th gadget of $G_{M2}^j$, no matter if it is a $G_0$ or $G_1$ gadget. Finally, sub-strings $\BB^4$ and $\E^4$ can have a match in the $G_{\BB\E}$ gadgets at the beginning and end of $G_{M2}^j$, respectively. All characters of $Q_i$ have now a matching node and the definition of the edges allows to visit all such nodes via a matching path starting in the left $G_{\BB\E}$ gadget of $G_{M2}^j$ and ending in the right $G_{\BB\E}$ gadget of $G_{M2}^j$.
\end{proof}

\begin{restatable}{lemma}{lemmapatternsubpattern}
\label{lemma:patternsubpattern}
String $Q$ has a match in $G$ if and only if a sub-string $Q_i$ of $Q$ has a match in the underlying sub-graph $G_{M2}$ of $G_M$.
\end{restatable}

\begin{proof}
For the $(\Rightarrow)$ implication, because of the directed $\E^4\BB^4$-edges, each distinct sub-string $Q_i$ matches a path from a distinct portion of either $G_L$, $G_M$ and $G_R$. Moreover, each occurrence of $P$ must begin with $\BB^8$ and end with $\E^8$. String $\BB^8$ can be matched only in $G_L$, in $G_{M1}$ or in $G_{M2}$, hence the match must start here. On the other hand, string $\E^8$ is found either in $G_{M2}$, $G_{M3}$ or in $G_R$. Observe that, by construction, once a match for pattern $Q$ is started in $G_L$, in $G_{M1}$ or in $G_{M2}$, the only way to successfully conclude it is either by matching $\E^8$ within $G_{M2}$, or by matching also a portion of $G_{M3}$ and/or $G_R$ and then $\E^8$.
Because of the structure of the graph, in both cases a sub-string $Q_i$ of $Q$ must match one of the gadget $G_{M2}^j$ that are present in $G_{M2}$.

The $(\Leftarrow)$ implication is trivial. In fact, if $Q_i$ has a match in one gadget $G_{M2}^j$, then by construction we can match $\BB^4 Q_1 \ldots Q_{i-1}$ possibly in $G_L$, then possibly in $G_{M1}$. We can then match $Q_{i+1} \ldots Q_n\E^4$ possibly in $G_{M3}$, then possibly in $G_R$, and thus have a full match for $Q$ in $G$.
\end{proof}

Our first lower bound is on matching a query string in an \efg without indexing.
\begin{restatable}{theorem}{thmefgraphsonlinelb}
\label{thm:efgraphs_onlinelb}
For any constant $\epsilon > 0$, it is not possible to find a match for a query string $Q$ into an \efg $G = (V,E,\ell)$ in either $O(|E|^{1-\epsilon} \, |Q|)$ or $O(|E| \, |Q|^{1-\epsilon})$ time, unless \ovh fails. This holds even if restricted to an alphabet of size $4$.
\end{restatable}

\begin{proof}
First, notice that the reduction that we presented for query string $Q$ and \efg $G$ is correct. Indeed, Lemma~\ref{lemma:patternsubpattern} guarantees that $Q$ has a match in $G$ if and only if a sub-string $Q_i$ has a match in $G_M$, and this holds, by Lemma~\ref{lemma:matchingGW}, if and only if $x_i \cdot y_j = 0$. Thus, string $Q$ has a match in $G$ if and only if there exist vectors $x_i \in X$ and $y_j \in Y$ which are orthogonal. 

The reduction requires linear time and space in the size $O(nd)$ of the \ov problem, and this is because of the construction of string $Q$ and graph $G$. On one hand, when we define string $Q$, we place a constant number of characters for each entry of each vector, thus $|Q| = O(nd)$. On the other hand, sub-graphs $G_L$, $G_{M1}$, $G_{M2}$, $G_{M3}$ and $G_R$ all consist of $O(n)$ structures, each one containing $O(d)$ nodes, and a constant number of edges for each node, for an overall size of 
$O(nd)$.

Hence, given two sets of vectors $X$ and $Y$, we can perform our reduction obtaining string $Q$ and \efg $G = (V,E,\ell)$ in $O(nd)$ time, while observing that $|E| = O(nd)$ and $|Q| = O(nd)$. If we can find a match for $Q$ in $G$ in $O(|E|^{1-\epsilon} |Q|)$ or $O(|E| \, |Q|^{1-\epsilon})$ time, then we can decide if there exists a pair of orthogonal vectors between $X$ and $Y$ in $O(nd \cdot (nd)^{1-\epsilon}) = O(n^{2-\epsilon}\text{poly}(d))$ time, which contradicts \ovh.
\end{proof}

We obtain the indexing lower bound by proving that the above reduction is a \emph{linear independent-components} (\emph{lic}) reduction, as defined by \cite[Definition~$3$]{EMT21}. 

\begin{restatable}{theorem}{thmefgraphsindexinglb}
\label{thm:efgraphs_indexinglb}
For any $\alpha,\beta,\delta>0$ such that $\beta+\delta<2$, there is no algorithm preprocessing an \efg $G = (V,E,\ell)$ in time $O(|E|^\alpha)$ such that for any query string $Q$ we can find a match for $Q$ in $G$ in time $O(|Q|+|E|^\delta|Q|^\beta)$, unless \ovh is false. This holds even if restricted to an alphabet of size $4$.
\end{restatable}
\begin{proof}
It is enough to notice that the reduction from \ov that we presented is a \emph{lic} reduction. Namely, (1) the reduction is correct and can be performed in linear time and space $O(nd)$ (recall the proof of Theorem~\ref{thm:efgraphs_onlinelb}), and (2) query string $|Q|$ is defined using only vector set $X$ and it is independent from vector set $Y$, while elastic founder graph $G$ is built using only vector set $Y$ and it is independent from vector set $X$. Hence, Corollary~$1$ in~\cite{EMT21} can be applied, proving our thesis.
\end{proof}

\section{Indexing \efgs\label{sect:gaps}}

Let us now consider how to extend the indexing results to the general case of \msas with gaps.
The idea is that gaps are only used in the segmentation algorithm to define the valid ranges, and that is the only place where special attention needs to be taken; elsewhere, whenever a substring from \msa rows is read, gaps are treated as empty strings. That is, A-GC-TA- becomes AGCTA.  

As we later see, segmentation becomes more difficult with gaps, and we need to consider a relaxed variant of prefix-free property for obtaining efficient algorithms. Recall that in a semi-repeat-free \efg a node label can appear as a prefix of another node label inside the same block.

\subsection{Repeat-free case}

As the reader can check, the indexing solutions in Sections~\ref{sect:repeat-freeness}~and~\ref{sect:compressedindexing} work verbatim with the repeat-free elastic founder graphs; the property of having equal-length of strings inside the blocks is not exploited in the algorithms.

\subsection{Semi-repeat-free case}

The case of semi-repeat-free elastic founder graphs is slightly more complex, and we need to combine and extend the previous solutions. The following lemma is the key property needed for the solution.

\begin{lemma}\label{lemma:keyproperty}
Consider a semi-repeat free \efg $G=(V,E)$. 
String $\ell(v)\ell(w)$, where $(v,e)\in E$, can only appear in $G$ as a prefix of paths starting with $v$.
\end{lemma}
\begin{proof}
Assume for contradiction that $\ell(v)\ell(w)$ is a prefix of a path starting inside the label of some $v'\in V$, $v'\neq v$. Then $\ell(v)$ is a prefix of such path, and this is only possible if $v'$ is in the same block as $v$ and $\ell(v)$ is a proper prefix of $\ell(v')$: otherwise $G$ would not be semi-repeat free. Then $|\ell(v)|<|\ell(v')|$ and $\ell(w)$ has an occurrence in a path starting inside the label $\ell(v')$. This is a contradiction on the fact that $G$ is semi-repeat free. 
\end{proof}

Consider the suffix tree of the concatenation $D=\prod_{v,w,u: (v,w)\in E, (w,u)\in E} (\ell(v)\ell(w)\ell(u))^{-1}\mathbf{0}$.
For suffixes of type $\alpha \ell(w)^{-1} \ell(v)^{-1}\mathbf{0}$, where $\alpha$ is a (possibly empty) string, we store node identifier $v$ in the corresponding leaf of the tree. Clearly, queries spanning less than three nodes can be located from this suffix tree. Consider a longer query $Q[1..q]$ whose suffix spans at least three nodes in the graph. We search it backwards in the suffix tree until reaching a locus after which we cannot proceed with $Q[i]$, but could continue with $\mathbf{0}$. Then we know that $Q[i+1..q]$ matches a path starting with $\ell(v)\ell(w)$ in $G$. Due to Lemma~\ref{lemma:keyproperty}, any leaf in the subtree rooted at the current locus in the suffix tree (which is spelling $Q[i+1..q]^{-1}$) stores $v$. Since we cannot know in advance if $Q$ is a longer query, we have stored identifiers $v$ only when this case applies. Once we have identified $v$, we only need to check if we can read $Q[1..i]$ following a path to the left from $v$, which is exactly what we did in Sect.~\ref{sect:repeat-freeness} using tries $\mathcal{R}(v)$ storing $\{\ell(u)^{-1} \mid (u,v)\in E\}$: The semi-repeat-free property guarantees that no node label can be a suffix of another node label (even inside the same block), and hence the leaves of the tries $\mathcal{R}(v)$ correspond to exactly one row each. The left-extensions are hence not branching, as the search always narrows down to one row (leaf), before continuing on the next trie (see Sect.~\ref{sect:repeat-freeness}). 

\begin{theorem}
A (semi-)repeat-free founder/block graph $G=(V,E)$ or a (semi-)repeat-free elastic degenerate string can be indexed in polynomial time into a data structure occupying $O(|D| \log |D|)$ bits of space, where $|D|=O(NH^2)$, $N$ is the total length of the node labels, and $H$ is the height of $G$. Later, one can find out in $O(|Q|)$ time if a given query string $Q$ occurs in $G$. 
\label{thm:indexing}
\end{theorem}
\begin{proof}
Each node label $\ell(v)$ is added to $D$ at most $3 H^2$ times, as $H$ is an upper bound for the number of edges from and to $v$.
The length of $D$ is then bounded by $O(NH^2)$.
Construction of suffix tree on $D$ can be done in linear time~\cite{Farach97}. In polynomial time, the nodes of the suffix tree can be preprocessed with perfect hash functions, such that following a downward path takes constant time per step. 
\end{proof}

We note that the index can be modified to report only matches that are (gap-oblivious) substrings of the \msa rows: Short patterns spanning only one edge are already such. Longer patterns can have only one occurrence in $G$, and it suffices to verify them with a regular string index on the \msa. Such modified scheme makes the approach functionality equivalent with wide range of indexes designed for repetitive collections \cite{MNSV09jcb,Naetal13a,Naetal13b,Naetal16a,Naetal16b,GN19,GNP20} and shares the benefit of alignment-based indexes of Na et al. \cite{Naetal13a,Naetal13b,Naetal16a,Naetal16b} in reporting the aligned matches only once, where e.g. r-index \cite{GNP20} needs to report all occurrences.     

Using compressed suffix trees, different space-time tradeoffs can be achieved. In Section~\ref{sect:wheeler}, we develop an alternative compressed indexing scheme for the repeat-free case using Wheeler graphs. 

\section{Construction of (semi-)repeat-free \efgs\label{sect:efgconstruction}}

Now that we have seen that (semi-)repeat-free \efgs are easy to index, it remains to consider their construction. First, we observe that the algorithms in Sect.~\ref{sect:construction} do not work verbatim: Theorem~\ref{thm:linearsegmentation} is based on Eq.~(\ref{eq:valid-segmentation-vj}), but now this recurrence is no longer valid, as left-extension of a valid block may not be a valid block. A counterexample is shown in Table~\ref{table:elasticvalid}.
On the other hand, Algorithms~\ref{algo:maxblocks}~and~\ref{algo:minmaxlength} use the right-extension property of a valid block, and this holds even with general \msas.

\begin{observation}
If segment $\msaij{1..m}{j+1..f(j)}$ is (semi-)repeat-free, then segment $\msaij{1..m}{j+1..j'}$ is (semi-)repeat-free for all $j'$ such that $f(j)<j'\leq n$.
\label{obs:right-extension}
\end{observation}

However, these algorithms assume we have precomputed for each $j$ the smallest integer $f(j)$ such that $\msaij{1..m}{j+1..f(j)}$ is a (semi-)repeat-free segment. The earlier sliding windows preprocessing algorithm for values $f(j)$ inherently assumes the values are monotonic (as is the case with gapless \msas), but this does not hold in the general case: Consider again Table~\ref{table:elasticvalid}. Let $j$ be the column just before the longer segment of \msa. Then $f(j)>|X|+1$ and $f(j+1)=|X|$.

\begin{table}
\centering
\caption{Semi-repeat-free segment and its extension (to the left) into a non-valid segment. Here the distinct strings $X$, $Y$, and $Z$, $|X|=|Y|$ do not appear elsewhere in \msa, except $Z$ is a prefix of $Y$. The longer segment is non-valid, because $X$ is not a prefix but a suffix of \texttt{A}$X$. Reversing the definition does not help, as the same segment contains \texttt{A}$Z$ as prefix of \texttt{A}$Y$. \label{table:elasticvalid}}
\begin{tabular}{c|l}
segment of \msa & longer segment of \msa \\
$X$ & \texttt{A}$X$\\
$X$ & \texttt{-}$X$\\
$Y$ & \texttt{A}$Y$\\
$Z$\texttt{-}$^{|Y|-|Z|}$ & \texttt{A}$Z$\texttt{-}$^{|Y|-|Z|}$\\
\end{tabular}
\end{table}

In order to be able to use Algorithms~\ref{algo:maxblocks}~and~\ref{algo:minmaxlength}, we derive a new preprocessing algorithm for values $f(j)$ that does not assume monotonicity.

\subsection{Preprocessing for the non-monotonic case}

\begin{figure*}
    \centering
    \includegraphics[scale=0.4]{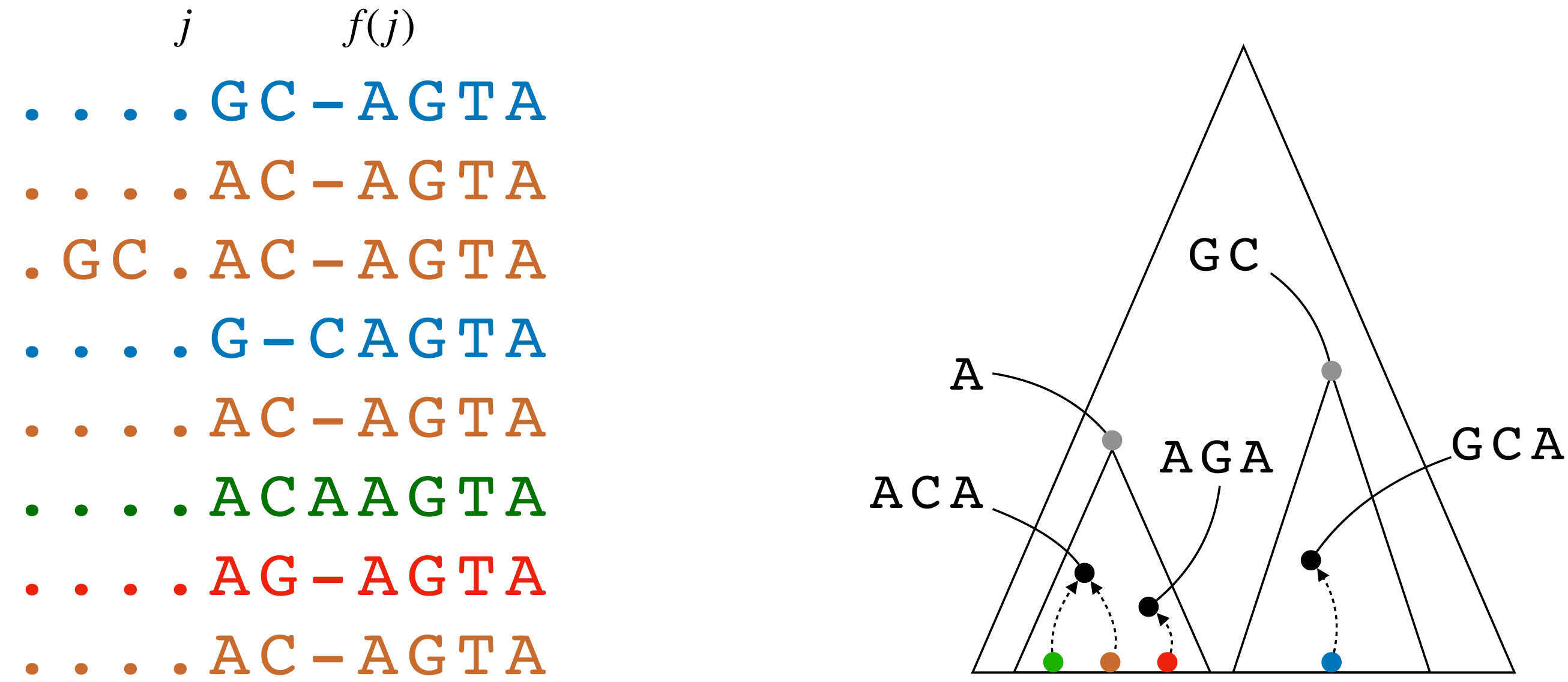}
    \caption{Illustrating the $O(m \log m)$ time algorithm to compute value $f(j)$ for a given $j$. Node labels correspond to the string spelled from the root to the node. We assume \texttt{ACA}, \texttt{AGA}, and \texttt{GCA} only appear in the region of the \msa visualized, while \texttt{GC} and \texttt{A} appear also elsewhere.}
    \label{fig:example-step}
\end{figure*}

As an internal part of the algorithm we need an efficient data structure to maintain a dynamic set of non-overlapping intervals. Let $I$ be a set of integer intervals, i.e., $I=\{[a_1..b_1],[a_2..b_2],\ldots, [a_m..b_m]\}$, where $a_i,b_i \in [1..n]$ for all $i$. We say $I$ is \emph{non-overlapping} if for all pairs $[a_i..b_i],[a_j..b_j] \in I$ holds $[a_i..b_i] \cap [a_j..b_j] = \emptyset$. 

\begin{lemma}
There is a data structure to maintain a non-overlapping set of intervals $I$ supporting insertions and deletions of intervals in $O(\log |I|)$ time. The data structure also supports in $O(\log |I|)$ time a query $\mathtt{span}([a..b])$ that returns $|\cup\{[a_i..b_i] \in I \mid a\leq a_i\leq b_i\leq b\}|$.
\label{lemma:dynamic-non-overlapping-intervals}
\end{lemma}
\begin{proof}
Consider a balanced binary search tree with leaves corresponding to intervals of $I$, sorted by values $a_i$ for $[a_i..b_i] \in I$. Leaf $i$ stores $a_i$ as \emph{key} value and $b_i-a_i+1$ as \emph{span} value. Each internal node $v$ stores the maximum  key and sum of span values of the leaves in its subtree. Assuming the data structure has been maintained with these values, answering query $\mathtt{span}([a..b])$ can be done as follows: Locate the $O(\log |I|)$ internal nodes that form a non-overlapping cover on the keys in the range $[a..b]$ (by searching keys $a$ and $b$ and picking the subtrees bypassed and within the interval). Return the sum of span values stored in those nodes.

It remains to consider how the values can be maintained during insertions and deletions, and during the resulting rebalancing operations. On each such operation, there are $O(\log |I|)$ internal nodes affected on the upward path from the the leaf to the root. It is sufficient to consider one such affected node $v$ assuming its left child $\ell$ and right child $r$ have already been updated accordingly. We set $\mathtt{key}(v)=\mathtt{key}(r)$ and $\mathtt{span}(v)=\mathtt{span}(\ell)+\mathtt{span}(r)$.
\end{proof}

\begin{restatable}{lemma}{fjlemma}
Let $f(j)$ be the smallest integer such that $\msaij{1..m}{j+1..f(j)}$ is a semi-repeat-free segment. We can compute all values $f(j)$ in $O(mn \log m)$ time. 
\label{lemma:non-monotonic-preprocessing}
\end{restatable}
\begin{proof}
Figure \ref{fig:example-step} illustrates the algorithm to be described in the following.
Consider the generalized suffix tree of $\{\spell(\msaij{i}{1..n}) \mid 1\leq i\leq m\}$. For each $j$, locate the subset $W$ of (implicit) suffix tree nodes corresponding to $\{\spell(\msaij{i}{j+1..n}) \mid 1\leq i\leq m\}$; ; these are the colored nodes in Fig.~\ref{fig:example-step}. If the number of leaves covered by the subtrees rooted at $W$ is greater than $m$, $f(j)$ remains undefined. 

Otherwise, we know that $f(j)\leq n$, and our aim is to decrease the right boundary, starting with $n$, until we have reached column $f(j)$. We do this one row at a time, recording values $f^i(j)$ such that in the end $f(j)=\max_i f^i(j)$. We initialize $f^i(j)=n$ for all $i$. Suffix tree nodes corresponding to rows whose values $f^i(j)$ are final are stored in set $F$, initially empty. Suffix tree nodes corresponding to rows whose values $f^i(j)$ are redundant (to be detailed later) are stored in set $R$, initially empty. 

To start this process, a) pick an (implicit) suffix tree node $v \in W$ corresponding to $\spell(\msaij{i}{j+1..f^i(j)})$ for some $i$. Let $w$ be the parent of $v$. It corresponds to string $\spell(\msaij{i}{j+1..j'})$ for some $j'$. Then b) consider replacing $v$ with $w$ in $W$. If the number of leaves covered by the subtrees rooted at $W \cup F \cup R$ is still $m$, then this replacement is safe, and we can set $f^i(j)=j'$. Safe replacements are shown as black nodes in Fig.~\ref{fig:example-step}, while the gray nodes are unsafe replacements. We also move from $W$ to $R$ all $w'\in W$ that are located in the subtree rooted at $w$ (not including $w$), as these nodes are now redundant; we will consider later how to compute their values $f^i(j)$. Otherwise, instead of replacement, move $v$ from $W$ to $F$, as we have found the minimum valid range with regards to row $i$: Consider string $\spell(\msaij{i}{j+1..j'+k})$, where $\msaij{i}{j'+k}$ is the first non-gap symbol at row $i$ after $\msaij{i}{j'}$. This string is spelled by reading the path from the root to $w$ and then reading one symbol on the edge $(w,v)$. This string is thus the shortest string having the same occurrences as $\spell(\msaij{i}{j+1..f^i(j)})$, and we can safely assign as final value $f^i(j)=j'+k$. 
Repeat these steps a) and b) until $V'$ is empty. At that point, decreasing of the right boundary is no longer possible on any row. 
However, we only have computed $f^i(j)$ for $i$ such that there is $v\in F$ corresponding to $\spell(\msaij{i}{j+1..f^i(j)})$.
We also need to compute values $f^{i'}(j)$ for $i'$ such that there is $v'\in R$ corresponding to $\spell(\msaij{i'}{j+1..f^{i'}(j)})$. Note that each $v'\in R$ was made redundant by another node, which in turn, may have been made redundant on its turn. We can store these relationships as a forest of trees. Root of each tree corresponds to some $v \in F$ and rest of the nodes are from $R$. Now, consider a root $v\in F$ of some of the trees corresponding to $\spell(\msaij{i}{j+1..f^i(j)})$ and a node $v' \in R$ of the same tree. We can assign $f^{i'}(j)=j'$ for smallest $j'$ such that $|\spell(\msaij{i'}{j+1..j'})|=|\spell(\msaij{i}{j+1..f^i(j)})|$. E.g.~for row $i=2$ in Fig.~\ref{fig:example-step}, we have $f^{2}(j)=j+3$, as $|\spell(\mathtt{AC-A})|=3=|\mathtt{ACA}|$.
Then we can set set $f(j)=\max_i f^i(j)$.

To achieve the claimed running time, we use backward searching on the unidirectional BWT index~\cite{BCKM20} on the concatenation of strings  $\{\spell(\msa[i,$ $1..n])$ $ \mid 1\leq i\leq m\}$ (with special markers added between) to find all the suffix array intervals corresponding to sets $\{\spell(\msaij{i}{j+1..n}) \mid 1\leq i\leq m\}$ for all $j$. This takes $O(mn)$ time. To find the largest $j$ for which the union of suffix array intervals is of size $m$, we can sort the intervals at each column and compute the size of the union by a simple scanning. This takes $O(mn \log m)$ time overall. 

Now we need to show that the process of reducing the right boundary for a fixed column can be done in $O(m \log m)$ time. Mapping from a suffix array interval to the (compressed) suffix tree node takes constant time~\cite{Sad07}. Steps a) and b) are repeated at most $m$ times at any column $j$: Either some row $i$ gets completed, or at least one row becomes redundant. In both cases, size of $W$ decreases at each step. The most time consuming part in this process is to compute the number of leaves in the union of subtrees. We can do this in $O(m \log m)$ time, by mapping the nodes back to suffix array intervals, and then computing the size of the union of intervals as above. However, we can only afford to do this at the first step of the process. For the rest of the steps we use Lemma~\ref{lemma:dynamic-non-overlapping-intervals}: To be able to use the lemma, we need to ensure only non-overlapping intervals are stored in the data structure. Thus, at the first step we remove duplicates and intervals that are nested in another one in $O(m \log m)$ time, and store the remaining intervals to the structure of Lemma~\ref{lemma:dynamic-non-overlapping-intervals}. While doing so, we move the suffix tree nodes corresponding to these removed intervals to the set $R$ of redundant nodes. This is safe, as initially the union of the intervals is of size $m$ (no extra occurrences in the intervals), and hence the steps a) and b) would anyway move the suffix tree nodes corresponding to those intervals to $R$ at some point. Consider now step a) with $w$ being parent of suffix tree node $v$. Let $[a..b]$ be the suffix array interval corresponding to $w$. We can query $\mathtt{span}([a..b])$ from the data structure, and if the answer is $m$, we remove the intervals in the query range, and insert $[a..b]$ in their place. 

It remains to consider how to find the first non-gap symbol $\msaij{i}{j'+k}$ at row $i$ after $\msaij{i}{j'}$, and how to find the smallest $j'$ such that $|\spell(\msaij{i'}{j+1..j'})|=x$ given $x$. These can be done in constant time after $O(mn)$ time preprocessing for rank and select queries on bitvectors marking locations of the the gap symbols. \end{proof}
 
\begin{corollary}
    After an $O(mn \log m)$ time preprocessing, Algorithms~\ref{algo:maxblocks} and \ref{algo:minmaxlength} compute the scores $\mathtt{maxblocks}(n)=b$ and $\mathtt{minmaxlength}(n)=\max\limits_{i:1\leq i \leq b} L(S^i)$ of optimal semi-repeat-free segmentations $S^1,S^2,\ldots,S^b$ of $\msaij{1..m}{1..n}$ in $O(n)$ and $O(n \log \log n)$ time, respectively. The produced segmentations induce semi-repeat-free \efgs from a general \msa.
\end{corollary}

\subsection{Prefix-free \efgs} 

While Algorithms~\ref{algo:maxblocks} and \ref{algo:minmaxlength} would work also for the prefix-free case, it appears difficult to modify the preprocessing for the same.

Instead of separate preprocessing and dynamic programming to compute the score of an optimal segmentation, we proceed directly with the main recurrence. We consider only minimizing the maximum block length score, as for this score we can derive a non-trivial parameterized solution. 

Let $e(j')$ be the score of a minimum scoring repeat-free segmentation $S^1,S^2,\ldots,S^b$ of prefix $\msa[1..m,1..j']$, where the score is defined as $\max\limits_{i:1\leq i \leq b} L(S^i)$. Then  
\begin{equation}
e(j)=\min_{\begin{array}{c} j':0\leq j'< j,\\ \msa[1..m,j'+1..j] \\\text{is repeat-free segment},\\ e(j')<j'+1 \end{array}} \max(j-j',e(j')), \label{eq:repeatfreeelasticscore}
\end{equation}
with $e(0)$ initialized to $0$. When there is no valid segmentation for some $j$, $e(j)=j+1$.

To test for a valid range $[j'+1..j]$, we adjust the sliding window preprocessing algorithm of Sect.~\ref{sect:preprocessing} in order to integrate it with the computation of the recurrence as follows:
\begin{enumerate}
    \item A unidirectional BWT index \cite{BCKM20} is built on the \msa rows concatenated into one long string, after the gap characters are removed and some separator symbols added between the rows.
    \item The search on each row of $\msa$ is initiated for each $j$, decreasing $j'$ from $j-1$ to $0$. This means only left-extensions are required.
    \item When $\msa[i,j]$ is accessed to alter the BWT interval, the old interval is retained if $\msa[i,j]=\text{-}$. 
    \item Modification 3) can cause intervals to become nested (exactly when substring $\text{spell}(\msa[i',j'..j])$ becomes a prefix of $\text{spell}(\msa[i,j'..j])$), and this needs to be checked for the proper detection of valid ranges. 
\end{enumerate}

Recall that at any step of the preprocessing algorithm of Sect.~\ref{sect:preprocessing}, we had a non-nested set $I=\{[i'_a\ltdots i_a]\}_{a\in \{1,2 \ldots m\}}$ of intervals. We exploited the non-nestedness in the use of a bitvectors $M$ (marking suffixes of current column), $B$ (BWT interval beginning), and $E$ (BWT interval ending) to detect if $I$ contains only BWT intervals of suffixes of the current column of the \msa. This gave us the linear time algorithm. Now that the intervals can become nested, these bitvectors no longer work as intended. Instead we resort to a generic method to check nestedness, and to compute the size of the union of distinct intervals in $I$, when no nestedness is detected. If no nestedness is detected, and the size of the union is $m$, we know that the range in consideration is valid. This can be done in $m \log m$ time e.g. by sorting the interval endpoints and simple scanning to maintain how many active intervals there are at the endpoints. If there is more than one active interval at any point, the range is not valid. Otherwise the range is valid, and the size of the union of intervals is just the sum of their lengths. This nestedness check and the computation of the union of intervals is repeated at each column.

\begin{theorem}
    The values $e(j)$ of Eq.~(\ref{eq:repeatfreeelasticscore}), for all $j \in [1 \ltdots n]$, can be computed in $O(mn e_{\mathtt{max}} \log m)$ time, where $e_{\mathtt{max}}=\max_{j} e(j)$. The optimal segmentation defined by Eq.~(\ref{eq:repeatfreeelasticscore}) yields a repeat-free elastic founder graph.
\label{thm:repeatfreeelasticscore}    
\end{theorem}
\begin{proof}
The unidirectional BWT index can be constructed in $O(mn)$ time and each left-extension takes constant time \cite{BCKM20}.
We can start comparing $\max(j-j',e(j'))$ from $j'=j-1$ decreasing $j'$ by one each step and maintaining $e(j)$ as the minimum value so far. Once value $j-j'$ grows bigger than current $e(j)$, we know that the value of $e(j)$ can no longer decrease. This means we can decrease $j'$ exactly $e(j)$ times. At each decrease of $j'$, we do $m$ left-extensions (one for each row), and then check for the validity by computing the union of search intervals in $O(m \log m)$ time. This gives the claimed running time. Traceback from $e(n)$ gives an optimal repeat-free segmentation.
\end{proof}

\section{Connection to Wheeler graphs\label{sect:wheeler}}

Wheeler graphs, also known as Wheeler automata, are a class of labeled graphs that admit an efficient index for path queries~\cite{GMS17}. We now give an alternative way to index repeat-free elastic block graphs by transforming the graph into an equivalent Wheeler automaton.

We view a block graph as a nondeterministic finite automaton (NFA) by adding a new initial state and edges from the source node to the starts of the first block, and expanding each string of each block to a path of states. To conform with automata notions, we define that the label of an edge is the label of the destination node. 

We denote the repeat-free NFA with $F$. First we determinize it with the standard subset construction for the reachable subsets of states. The states of the DFA are subsets of states of the NFA such that there is an edge from subset $S_1$ to subset $S_2$ with label $c$ iff $S_2$ is the set of states at the destinations of edges labeled with $c$ from  $S_1$. We only represent the subsets of states reachable from the subset containing only the initial state. We call the deterministic graph $G$. See Figures \ref{fig:NFA} and \ref{fig:DFA} for an example. 

A DFA is indexable as a Wheeler graph if there exists an order $<$ on the nodes such that if $u < v$, then every incoming path label to $u$ is colexicographically smaller than every incoming path label to $v$ (recall that the colexicographic order of strings is the lexicographic order of the reverses of the strings). The repeat-free property guarantees that the nodes at the ends of the blocks can be ordered among themselves by picking an arbitrary incoming path as the sorting key. 

To make sure that the rest of the nodes are sortable, we modify the graph so that if a node is not at the end of a block, we make it so that the incoming paths to the node do not branch backward before the backward path reaches the end of a previous block. This is done by turning each block into a set of disjoint trees, where the roots of the trees are the ending nodes of the previous block, in a way that preserves the language of the automaton. The roots may have multiple incoming edges from the leaves of the previous tree. See Figure \ref{fig:WDFA} for an example. The formal definition of the transformation and the proof of sortability are in Sect. \ref{sect:wheelerappendix}. We denote the transformed graph with $G'$ and obtain the following result:

\begin{restatable}{lemma}{WGsizerestatable}
The number of nodes in $G'$ is at most $O(NW)$, where $W$ is the maximum number of strings in a block of $F$ and $N$ is the total number of nodes in $F$.
\label{lemma:WG_size}
\end{restatable}

The Wheeler order $<$ of the transformed graph can be found by running the XBWT sorting algorithm on a spanning tree of the graph, as shown by Alanko et al.~\cite{Alaetal20}. Finally, we can find the minimum equivalent Wheeler graph by running the general Wheeler graph minimization algorithm of Alanko et al.~\cite{Alaetal20}.

With the input graph now converted into a Wheeler graph, one can deploy succinct data structures supporting fast pattern matching \cite[Lemma 4]{GMS17}, leading to the following result:

\begin{corollary}
A repeat-free founder/block graph $G$ or a repeat-free elastic degenerate string can be indexed in $O(NH)$ time into a Wheeler-graph-based data structure occupying $O(NH \log |\Sigma|)$ bits of space, where $N$ is the total number of characters in the node labels of $G$, $H$ is the height of $G$ (maximum number of strings in a block of $G$), and $\Sigma$ is the alphabet. Later, using the data structure, one can find out in $O(|Q|\log |\Sigma|)$ time if a given query string $Q$ occurs in $G$.
\label{cor:wheelerindexing}
\end{corollary}

\begin{figure}
    \centering
    
    \begin{tikzpicture}%
      [
       scale=0.75,
       every state/.style={scale = 1.0, draw=none, fill=none},
       scale=0.9, every node/.style={scale=0.9}
      ]  

      \node[draw, black] (27_label) at (0, 0) {\small{\$}};
\node[align=left, black] (27) at (0, 0.5) {\small{27}};
\node[draw, black] (0_label) at (1.7, 0) {\small{C}};
\node[align=left, black] (0) at (1.7, 0.5) {\small{0}};
\node[draw, black] (1_label) at (1.7, 1) {\small{T}};
\node[align=left, black] (1) at (1.7, 1.5) {\small{1}};
\node[draw, black] (2_label) at (1.7, 2) {\small{T}};
\node[align=left, black] (2) at (1.7, 2.5) {\small{2}};
\node[draw, black] (9_label) at (3.4, 0) {\small{A}};
\node[align=left, black] (9) at (3.4, 0.5) {\small{9}};
\node[draw, black] (10_label) at (3.4, 1) {\small{T}};
\node[align=left, black] (10) at (3.4, 1.5) {\small{10}};
\node[draw, black] (11_label) at (3.4, 2) {\small{G}};
\node[align=left, black] (11) at (3.4, 2.5) {\small{11}};
\node[draw, black] (3_label) at (5.1, 0) {\small{G}};
\node[align=left, black] (3) at (5.1, 0.5) {\small{3}};
\node[draw, black] (4_label) at (5.1, 1) {\small{G}};
\node[align=left, black] (4) at (5.1, 1.5) {\small{4}};
\node[draw, black] (5_label) at (5.1, 2) {\small{G}};
\node[align=left, black] (5) at (5.1, 2.5) {\small{5}};
\node[draw, black] (12_label) at (6.8, 0) {\small{C}};
\node[align=left, black] (12) at (6.8, 0.5) {\small{12}};
\node[draw, black] (15_label) at (6.8, 1) {\small{A}};
\node[align=left, black] (15) at (6.8, 1.5) {\small{15}};
\node[draw, black] (18_label) at (6.8, 2) {\small{A}};
\node[align=left, black] (18) at (6.8, 2.5) {\small{18}};
\node[draw, black] (13_label) at (8.5, 0) {\small{T}};
\node[align=left, black] (13) at (8.5, 0.5) {\small{13}};
\node[draw, black] (16_label) at (8.5, 1) {\small{T}};
\node[align=left, black] (16) at (8.5, 1.5) {\small{16}};
\node[draw, black] (19_label) at (8.5, 2) {\small{C}};
\node[align=left, black] (19) at (8.5, 2.5) {\small{19}};
\node[draw, black] (14_label) at (10.2, 0) {\small{G}};
\node[align=left, black] (14) at (10.2, 0.5) {\small{14}};
\node[draw, black] (17_label) at (10.2, 1) {\small{A}};
\node[align=left, black] (17) at (10.2, 1.5) {\small{17}};
\node[draw, black] (20_label) at (10.2, 2) {\small{T}};
\node[align=left, black] (20) at (10.2, 2.5) {\small{20}};
\node[draw, black] (7_label) at (11.9, 0) {\small{T}};
\node[align=left, black] (7) at (11.9, 0.5) {\small{7}};
\node[draw, black] (8_label) at (11.9, 1) {\small{T}};
\node[align=left, black] (8) at (11.9, 1.5) {\small{8}};
\node[draw, black] (6_label) at (11.9, 2) {\small{G}};
\node[align=left, black] (6) at (11.9, 2.5) {\small{6}};
\node[draw, black] (23_label) at (13.6, 0) {\small{A}};
\node[align=left, black] (23) at (13.6, 0.5) {\small{23}};
\node[draw, black] (25_label) at (13.6, 1) {\small{A}};
\node[align=left, black] (25) at (13.6, 1.5) {\small{25}};
\node[draw, black] (21_label) at (13.6, 2) {\small{C}};
\node[align=left, black] (21) at (13.6, 2.5) {\small{21}};
\node[draw, black] (24_label) at (15.3, 0) {\small{G}};
\node[align=left, black] (24) at (15.3, 0.5) {\small{24}};
\node[draw, black] (26_label) at (15.3, 1) {\small{A}};
\node[align=left, black] (26) at (15.3, 1.5) {\small{26}};
\node[draw, black] (22_label) at (15.3, 2) {\small{G}};
\node[align=left, black] (22) at (15.3, 2.5) {\small{22}};
\draw [line width=0.03cm, -, opacity=1.0, color=black] (2, 0) -- (3.1, 0);
\draw [line width=0.03cm, -, opacity=1.0, color=black] (2, 1) -- (3.1, 1);
\draw [line width=0.03cm, -, opacity=1.0, color=black] (2, 2) -- (3.1, 2);
\draw [line width=0.03cm, -, opacity=1.0, color=black] (5.4, 0) -- (6.5, 0);
\draw [line width=0.03cm, -, opacity=1.0, color=black] (5.4, 1) -- (6.5, 1);
\draw [line width=0.03cm, -, opacity=1.0, color=black] (5.4, 2) -- (6.5, 2);
\draw [line width=0.03cm, -, opacity=1.0, color=black] (12.2, 2) -- (13.3, 2);
\draw [line width=0.03cm, -, opacity=1.0, color=black] (12.2, 0) -- (13.3, 0);
\draw [line width=0.03cm, -, opacity=1.0, color=black] (12.2, 1) -- (13.3, 1);
\draw [line width=0.03cm, -, opacity=1.0, color=black] (3.7, 0) -- (4.8, 0);
\draw [line width=0.03cm, -, opacity=1.0, color=black] (3.7, 0) -- (4.8, 1);
\draw [line width=0.03cm, -, opacity=1.0, color=black] (3.7, 0) -- (4.8, 2);
\draw [line width=0.03cm, -, opacity=1.0, color=black] (3.7, 1) -- (4.8, 1);
\draw [line width=0.03cm, -, opacity=1.0, color=black] (3.7, 1) -- (4.8, 2);
\draw [line width=0.03cm, -, opacity=1.0, color=black] (3.7, 2) -- (4.8, 1);
\draw [line width=0.03cm, -, opacity=1.0, color=black] (3.7, 2) -- (4.8, 2);
\draw [line width=0.03cm, -, opacity=1.0, color=black] (7.1, 0) -- (8.2, 0);
\draw [line width=0.03cm, -, opacity=1.0, color=black] (8.8, 0) -- (9.9, 0);
\draw [line width=0.03cm, -, opacity=1.0, color=black] (10.5, 0) -- (11.6, 0);
\draw [line width=0.03cm, -, opacity=1.0, color=black] (10.5, 0) -- (11.6, 1);
\draw [line width=0.03cm, -, opacity=1.0, color=black] (7.1, 1) -- (8.2, 1);
\draw [line width=0.03cm, -, opacity=1.0, color=black] (8.8, 1) -- (9.9, 1);
\draw [line width=0.03cm, -, opacity=1.0, color=black] (10.5, 1) -- (11.6, 0);
\draw [line width=0.03cm, -, opacity=1.0, color=black] (10.5, 1) -- (11.6, 1);
\draw [line width=0.03cm, -, opacity=1.0, color=black] (7.1, 2) -- (8.2, 2);
\draw [line width=0.03cm, -, opacity=1.0, color=black] (8.8, 2) -- (9.9, 2);
\draw [line width=0.03cm, -, opacity=1.0, color=black] (10.5, 2) -- (11.6, 2);
\draw [line width=0.03cm, -, opacity=1.0, color=black] (10.5, 2) -- (11.6, 0);
\draw [line width=0.03cm, -, opacity=1.0, color=black] (10.5, 2) -- (11.6, 1);
\draw [line width=0.03cm, -, opacity=1.0, color=black] (13.9, 2) -- (15, 2);
\draw [line width=0.03cm, -, opacity=1.0, color=black] (13.9, 0) -- (15, 0);
\draw [line width=0.03cm, -, opacity=1.0, color=black] (13.9, 1) -- (15, 1);
\draw [line width=0.03cm, -, opacity=1.0, color=black] (0.3, 0) -- (1.4, 0);
\draw [line width=0.03cm, -, opacity=1.0, color=black] (0.3, 0) -- (1.4, 1);
\draw [line width=0.03cm, -, opacity=1.0, color=black] (0.3, 0) -- (1.4, 2);
\fill[yellow, opacity=0.3] (2.89,-0.7) rectangle (3.91,2.7);
\fill[yellow, opacity=0.3] (9.69,-0.7) rectangle (10.71,2.7);
\fill[yellow, opacity=0.3] (14.79,-0.7) rectangle (15.81,2.7);
    \end{tikzpicture}
    \caption{Repeat-free block NFA. The last columns of each block are highlighted.}
    \label{fig:NFA}
  \end{figure}
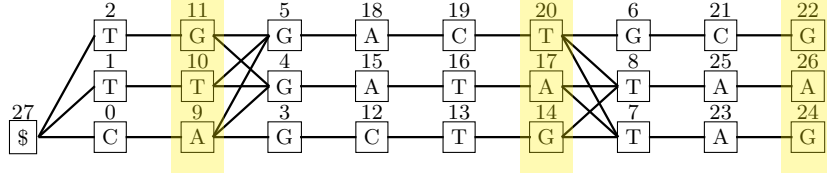
    
  \begin{figure}

    \centering
    
    \begin{tikzpicture}%
      [
       scale=0.75,
       every state/.style={scale = 1.0, draw=none, fill=none},
       scale=0.9, every node/.style={scale=0.9}
      ]  

      \node[draw, black] (0_label) at (0, 0) {\small{\$}};
\node[align=left, black] (0) at (0, 0.5) {\small{27}};
\node[draw, black] (1_label) at (1.7, 0) {\small{C}};
\node[align=left, black] (1) at (1.7, 0.5) {\small{0}};
\node[draw, black] (2_label) at (1.7, 1) {\small{T}};
\node[align=left, black] (2) at (1.7, 1.5) {\small{1,2}};
\node[draw, black] (3_label) at (3.4, 0) {\small{A}};
\node[align=left, black] (3) at (3.4, 0.5) {\small{9}};
\node[draw, black] (4_label) at (3.4, 1) {\small{G}};
\node[align=left, black] (4) at (3.4, 1.5) {\small{11}};
\node[draw, black] (5_label) at (3.4, 2) {\small{T}};
\node[align=left, black] (5) at (3.4, 2.5) {\small{10}};
\node[draw, black] (6_label) at (5.1, 0) {\small{G}};
\node[align=left, black] (6) at (5.1, 0.5) {\small{3,4,5}};
\node[draw, black] (7_label) at (5.1, 1) {\small{G}};
\node[align=left, black] (7) at (5.1, 1.5) {\small{4,5}};
\node[draw, black] (8_label) at (6.8, 0) {\small{A}};
\node[align=left, black] (8) at (6.8, 0.5) {\small{15,18}};
\node[draw, black] (9_label) at (6.8, 1) {\small{C}};
\node[align=left, black] (9) at (6.8, 1.5) {\small{12}};
\node[draw, black] (10_label) at (8.5, 0) {\small{C}};
\node[align=left, black] (10) at (8.5, 0.5) {\small{19}};
\node[draw, black] (11_label) at (8.5, 1) {\small{T}};
\node[align=left, black] (11) at (8.5, 1.5) {\small{16}};
\node[draw, black] (12_label) at (8.5, 2) {\small{T}};
\node[align=left, black] (12) at (8.5, 2.5) {\small{13}};
\node[draw, black] (13_label) at (10.2, 0) {\small{T}};
\node[align=left, black] (13) at (10.2, 0.5) {\small{20}};
\node[draw, black] (14_label) at (10.2, 1) {\small{A}};
\node[align=left, black] (14) at (10.2, 1.5) {\small{17}};
\node[draw, black] (15_label) at (10.2, 2) {\small{G}};
\node[align=left, black] (15) at (10.2, 2.5) {\small{14}};
\node[draw, black] (16_label) at (11.9, 0) {\small{G}};
\node[align=left, black] (16) at (11.9, 0.5) {\small{6}};
\node[draw, black] (17_label) at (11.9, 1) {\small{T}};
\node[align=left, black] (17) at (11.9, 1.5) {\small{7,8}};
\node[draw, black] (18_label) at (13.6, 0) {\small{C}};
\node[align=left, black] (18) at (13.6, 0.5) {\small{21}};
\node[draw, black] (19_label) at (13.6, 1) {\small{A}};
\node[align=left, black] (19) at (13.6, 1.5) {\small{23,25}};
\node[draw, black] (20_label) at (15.3, 0) {\small{G}};
\node[align=left, black] (20) at (15.3, 0.5) {\small{22}};
\node[draw, black] (21_label) at (15.3, 1) {\small{A}};
\node[align=left, black] (21) at (15.3, 1.5) {\small{26}};
\node[draw, black] (22_label) at (15.3, 2) {\small{G}};
\node[align=left, black] (22) at (15.3, 2.5) {\small{24}};
\draw [line width=0.03cm, -, opacity=1.0, color=black] (0.3, 0) -- (1.4, 0);
\draw [line width=0.03cm, -, opacity=1.0, color=black] (0.3, 0) -- (1.4, 1);
\draw [line width=0.03cm, -, opacity=1.0, color=black] (2, 0) -- (3.1, 0);
\draw [line width=0.03cm, -, opacity=1.0, color=black] (2, 1) -- (3.1, 1);
\draw [line width=0.03cm, -, opacity=1.0, color=black] (2, 1) -- (3.1, 2);
\draw [line width=0.03cm, -, opacity=1.0, color=black] (3.7, 0) -- (4.8, 0);
\draw [line width=0.03cm, -, opacity=1.0, color=black] (3.7, 1) -- (4.8, 1);
\draw [line width=0.03cm, -, opacity=1.0, color=black] (3.7, 2) -- (4.8, 1);
\draw [line width=0.03cm, -, opacity=1.0, color=black] (5.4, 0) -- (6.5, 0);
\draw [line width=0.03cm, -, opacity=1.0, color=black] (5.4, 0) -- (6.5, 1);
\draw [line width=0.03cm, -, opacity=1.0, color=black] (5.4, 1) -- (6.5, 0);
\draw [line width=0.03cm, -, opacity=1.0, color=black] (7.1, 0) -- (8.2, 0);
\draw [line width=0.03cm, -, opacity=1.0, color=black] (7.1, 0) -- (8.2, 1);
\draw [line width=0.03cm, -, opacity=1.0, color=black] (7.1, 1) -- (8.2, 2);
\draw [line width=0.03cm, -, opacity=1.0, color=black] (8.8, 0) -- (9.9, 0);
\draw [line width=0.03cm, -, opacity=1.0, color=black] (8.8, 1) -- (9.9, 1);
\draw [line width=0.03cm, -, opacity=1.0, color=black] (8.8, 2) -- (9.9, 2);
\draw [line width=0.03cm, -, opacity=1.0, color=black] (10.5, 0) -- (11.6, 0);
\draw [line width=0.03cm, -, opacity=1.0, color=black] (10.5, 0) -- (11.6, 1);
\draw [line width=0.03cm, -, opacity=1.0, color=black] (10.5, 1) -- (11.6, 1);
\draw [line width=0.03cm, -, opacity=1.0, color=black] (10.5, 2) -- (11.6, 1);
\draw [line width=0.03cm, -, opacity=1.0, color=black] (12.2, 0) -- (13.3, 0);
\draw [line width=0.03cm, -, opacity=1.0, color=black] (12.2, 1) -- (13.3, 1);
\draw [line width=0.03cm, -, opacity=1.0, color=black] (13.9, 0) -- (15, 0);
\draw [line width=0.03cm, -, opacity=1.0, color=black] (13.9, 1) -- (15, 1);
\draw [line width=0.03cm, -, opacity=1.0, color=black] (13.9, 1) -- (15, 2);
\fill[yellow, opacity=0.3] (2.89,-0.7) rectangle (3.91,2.7);
\fill[yellow, opacity=0.3] (9.69,-0.7) rectangle (10.71,2.7);
\fill[yellow, opacity=0.3] (14.79,-0.7) rectangle (15.81,2.7);
    \end{tikzpicture}
    \caption{The DFA resulting from the subset construction for the NFA in Figure \ref{fig:NFA}. The numbers above the nodes specify the subset of NFA states corresponding to the DFA state.}
    \label{fig:DFA}
  \end{figure}
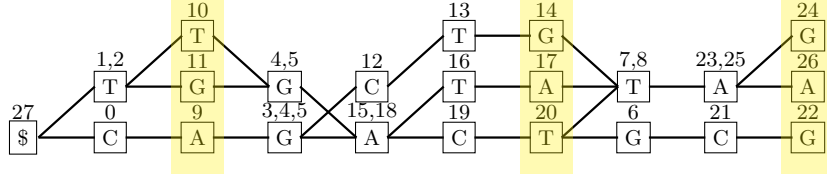

  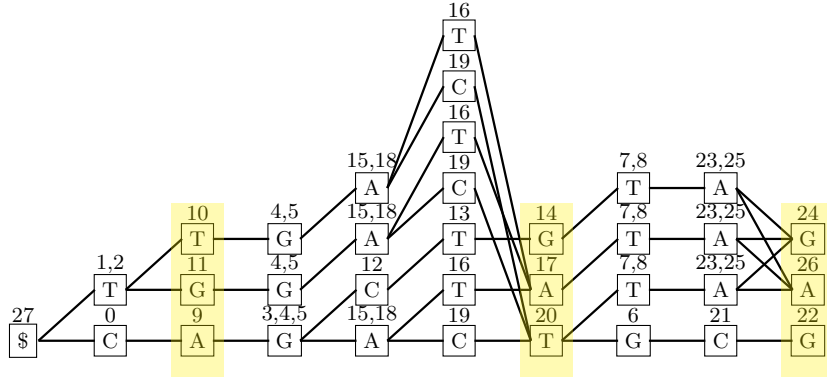
\begin{figure}

    \centering
    
    \begin{tikzpicture}%
      [
       scale=0.75,
       every state/.style={scale = 1.0, draw=none, fill=none},
       scale=0.9, every node/.style={scale=0.9}
      ]  

      \node[draw, black] (0_label) at (0, 0) {\small{\$}};
\node[align=left, black] (0) at (0, 0.5) {\small{27}};
\node[draw, black] (1_label) at (1.7, 0) {\small{C}};
\node[align=left, black] (1) at (1.7, 0.5) {\small{0}};
\node[draw, black] (2_label) at (1.7, 1) {\small{T}};
\node[align=left, black] (2) at (1.7, 1.5) {\small{1,2}};
\node[draw, black] (3_label) at (3.4, 0) {\small{A}};
\node[align=left, black] (3) at (3.4, 0.5) {\small{9}};
\node[draw, black] (4_label) at (3.4, 1) {\small{G}};
\node[align=left, black] (4) at (3.4, 1.5) {\small{11}};
\node[draw, black] (5_label) at (3.4, 2) {\small{T}};
\node[align=left, black] (5) at (3.4, 2.5) {\small{10}};
\node[draw, black] (6_label) at (5.1, 0) {\small{G}};
\node[align=left, black] (6) at (5.1, 0.5) {\small{3,4,5}};
\node[draw, black] (7_label) at (5.1, 1) {\small{G}};
\node[align=left, black] (7) at (5.1, 1.5) {\small{4,5}};
\node[draw, black] (8_label) at (5.1, 2) {\small{G}};
\node[align=left, black] (8) at (5.1, 2.5) {\small{4,5}};
\node[draw, black] (9_label) at (6.8, 0) {\small{A}};
\node[align=left, black] (9) at (6.8, 0.5) {\small{15,18}};
\node[draw, black] (10_label) at (6.8, 1) {\small{C}};
\node[align=left, black] (10) at (6.8, 1.5) {\small{12}};
\node[draw, black] (11_label) at (6.8, 2) {\small{A}};
\node[align=left, black] (11) at (6.8, 2.5) {\small{15,18}};
\node[draw, black] (12_label) at (6.8, 3) {\small{A}};
\node[align=left, black] (12) at (6.8, 3.5) {\small{15,18}};
\node[draw, black] (13_label) at (8.5, 0) {\small{C}};
\node[align=left, black] (13) at (8.5, 0.5) {\small{19}};
\node[draw, black] (14_label) at (8.5, 1) {\small{T}};
\node[align=left, black] (14) at (8.5, 1.5) {\small{16}};
\node[draw, black] (15_label) at (8.5, 2) {\small{T}};
\node[align=left, black] (15) at (8.5, 2.5) {\small{13}};
\node[draw, black] (16_label) at (8.5, 3) {\small{C}};
\node[align=left, black] (16) at (8.5, 3.5) {\small{19}};
\node[draw, black] (17_label) at (8.5, 4) {\small{T}};
\node[align=left, black] (17) at (8.5, 4.5) {\small{16}};
\node[draw, black] (18_label) at (8.5, 5) {\small{C}};
\node[align=left, black] (18) at (8.5, 5.5) {\small{19}};
\node[draw, black] (19_label) at (8.5, 6) {\small{T}};
\node[align=left, black] (19) at (8.5, 6.5) {\small{16}};
\node[draw, black] (20_label) at (10.2, 0) {\small{T}};
\node[align=left, black] (20) at (10.2, 0.5) {\small{20}};
\node[draw, black] (21_label) at (10.2, 1) {\small{A}};
\node[align=left, black] (21) at (10.2, 1.5) {\small{17}};
\node[draw, black] (22_label) at (10.2, 2) {\small{G}};
\node[align=left, black] (22) at (10.2, 2.5) {\small{14}};
\node[draw, black] (23_label) at (11.9, 0) {\small{G}};
\node[align=left, black] (23) at (11.9, 0.5) {\small{6}};
\node[draw, black] (24_label) at (11.9, 1) {\small{T}};
\node[align=left, black] (24) at (11.9, 1.5) {\small{7,8}};
\node[draw, black] (25_label) at (11.9, 2) {\small{T}};
\node[align=left, black] (25) at (11.9, 2.5) {\small{7,8}};
\node[draw, black] (26_label) at (11.9, 3) {\small{T}};
\node[align=left, black] (26) at (11.9, 3.5) {\small{7,8}};
\node[draw, black] (27_label) at (13.6, 0) {\small{C}};
\node[align=left, black] (27) at (13.6, 0.5) {\small{21}};
\node[draw, black] (28_label) at (13.6, 1) {\small{A}};
\node[align=left, black] (28) at (13.6, 1.5) {\small{23,25}};
\node[draw, black] (29_label) at (13.6, 2) {\small{A}};
\node[align=left, black] (29) at (13.6, 2.5) {\small{23,25}};
\node[draw, black] (30_label) at (13.6, 3) {\small{A}};
\node[align=left, black] (30) at (13.6, 3.5) {\small{23,25}};
\node[draw, black] (31_label) at (15.3, 0) {\small{G}};
\node[align=left, black] (31) at (15.3, 0.5) {\small{22}};
\node[draw, black] (32_label) at (15.3, 1) {\small{A}};
\node[align=left, black] (32) at (15.3, 1.5) {\small{26}};
\node[draw, black] (33_label) at (15.3, 2) {\small{G}};
\node[align=left, black] (33) at (15.3, 2.5) {\small{24}};
\draw [line width=0.03cm, -, opacity=1.0, color=black] (0.3, 0) -- (1.4, 0);
\draw [line width=0.03cm, -, opacity=1.0, color=black] (0.3, 0) -- (1.4, 1);
\draw [line width=0.03cm, -, opacity=1.0, color=black] (2, 0) -- (3.1, 0);
\draw [line width=0.03cm, -, opacity=1.0, color=black] (2, 1) -- (3.1, 1);
\draw [line width=0.03cm, -, opacity=1.0, color=black] (2, 1) -- (3.1, 2);
\draw [line width=0.03cm, -, opacity=1.0, color=black] (3.7, 0) -- (4.8, 0);
\draw [line width=0.03cm, -, opacity=1.0, color=black] (3.7, 1) -- (4.8, 1);
\draw [line width=0.03cm, -, opacity=1.0, color=black] (3.7, 2) -- (4.8, 2);
\draw [line width=0.03cm, -, opacity=1.0, color=black] (5.4, 0) -- (6.5, 0);
\draw [line width=0.03cm, -, opacity=1.0, color=black] (5.4, 0) -- (6.5, 1);
\draw [line width=0.03cm, -, opacity=1.0, color=black] (5.4, 1) -- (6.5, 2);
\draw [line width=0.03cm, -, opacity=1.0, color=black] (5.4, 2) -- (6.5, 3);
\draw [line width=0.03cm, -, opacity=1.0, color=black] (7.1, 0) -- (8.2, 0);
\draw [line width=0.03cm, -, opacity=1.0, color=black] (7.1, 0) -- (8.2, 1);
\draw [line width=0.03cm, -, opacity=1.0, color=black] (7.1, 1) -- (8.2, 2);
\draw [line width=0.03cm, -, opacity=1.0, color=black] (7.1, 2) -- (8.2, 3);
\draw [line width=0.03cm, -, opacity=1.0, color=black] (7.1, 2) -- (8.2, 4);
\draw [line width=0.03cm, -, opacity=1.0, color=black] (7.1, 3) -- (8.2, 5);
\draw [line width=0.03cm, -, opacity=1.0, color=black] (7.1, 3) -- (8.2, 6);
\draw [line width=0.03cm, -, opacity=1.0, color=black] (8.8, 0) -- (9.9, 0);
\draw [line width=0.03cm, -, opacity=1.0, color=black] (8.8, 1) -- (9.9, 1);
\draw [line width=0.03cm, -, opacity=1.0, color=black] (8.8, 2) -- (9.9, 2);
\draw [line width=0.03cm, -, opacity=1.0, color=black] (8.8, 3) -- (9.9, 0);
\draw [line width=0.03cm, -, opacity=1.0, color=black] (8.8, 4) -- (9.9, 1);
\draw [line width=0.03cm, -, opacity=1.0, color=black] (8.8, 5) -- (9.9, 0);
\draw [line width=0.03cm, -, opacity=1.0, color=black] (8.8, 6) -- (9.9, 1);
\draw [line width=0.03cm, -, opacity=1.0, color=black] (10.5, 0) -- (11.6, 0);
\draw [line width=0.03cm, -, opacity=1.0, color=black] (10.5, 0) -- (11.6, 1);
\draw [line width=0.03cm, -, opacity=1.0, color=black] (10.5, 1) -- (11.6, 2);
\draw [line width=0.03cm, -, opacity=1.0, color=black] (10.5, 2) -- (11.6, 3);
\draw [line width=0.03cm, -, opacity=1.0, color=black] (12.2, 0) -- (13.3, 0);
\draw [line width=0.03cm, -, opacity=1.0, color=black] (12.2, 1) -- (13.3, 1);
\draw [line width=0.03cm, -, opacity=1.0, color=black] (12.2, 2) -- (13.3, 2);
\draw [line width=0.03cm, -, opacity=1.0, color=black] (12.2, 3) -- (13.3, 3);
\draw [line width=0.03cm, -, opacity=1.0, color=black] (13.9, 0) -- (15, 0);
\draw [line width=0.03cm, -, opacity=1.0, color=black] (13.9, 1) -- (15, 1);
\draw [line width=0.03cm, -, opacity=1.0, color=black] (13.9, 1) -- (15, 2);
\draw [line width=0.03cm, -, opacity=1.0, color=black] (13.9, 2) -- (15, 1);
\draw [line width=0.03cm, -, opacity=1.0, color=black] (13.9, 2) -- (15, 2);
\draw [line width=0.03cm, -, opacity=1.0, color=black] (13.9, 3) -- (15, 1);
\draw [line width=0.03cm, -, opacity=1.0, color=black] (13.9, 3) -- (15, 2);
\fill[yellow, opacity=0.3] (2.89,-0.7) rectangle (3.91,2.7);
\fill[yellow, opacity=0.3] (9.69,-0.7) rectangle (10.71,2.7);
\fill[yellow, opacity=0.3] (14.79,-0.7) rectangle (15.81,2.7);
    \end{tikzpicture}
    \caption{The Wheeler DFA resulting from running our Wheeler expansion algorithm on the DFA in Figure \ref{fig:DFA}.}
    \label{fig:WDFA}
  \end{figure}

\subsection{Wheeler graph details and proofs \label{sect:wheelerappendix}}

Suppose we have the NFA $F$ corresponding to a repeat-free \efg, and let the $G$ be the determinization of $F$ defined above. Denote with $P(v)$ the set of path labels from the initial state to $v$. Since the graph is a DFA, all sets $P(v)$ are disjoint. Denote with $P_{min}(v)$ and $P_{max}(v)$ the colexicographic minimum and maximum of $P(v)$ respectively. We denote the colexicographic order relation with $\prec$. 

We say that a node $v$ is \emph{atomic} if for all path labels $\alpha$ in the graph, we have $\alpha \in P(v)$ iff $P_{min}(v) \preceq \alpha \preceq P_{max}(v)$. A DFA is Wheeler if and only if all its nodes are atomic~\cite{Alaetal20}. In this case, the Wheeler order $<$ of nodes is defined so that $v < u$ if $\alpha \prec \beta$ for some strings $\alpha \in P(v)$ and $\beta \in P(u)$ (this is well-defined when all nodes are atomic).

We will expand the graph so that all its nodes become atomic. To achieve this, we process the nodes in a topological order. If the current node $v$ is at the end of a block, we do nothing. Otherwise, we apply the following transformation. Suppose the in-neighbors of $v$ are $u_1,\ldots,u_k$ and the out-neighbors of $v$ are $w_1,\ldots, w_l$. We delete $v$, add $k$ new nodes $v_1,\ldots v_k$ and add the sets of edges $\{(u_i, v_i) \; | \; 1 \leq i \leq k \}$ and $\{(v_i,w_j) \; | \; 1 \leq i \leq k \textrm{ and } 1 \leq j \leq l\}$. In other words, we \emph{distribute} the in-edges of $v$ and \emph{duplicate} the out-edges. The automaton remains deterministic after this. Figure \ref{fig:WDFA} shows an example of the expansion.

We denote the resulting graph after all transformations with $G'$. First we note that by construction, the language of $G'$ is the same as the language of $G$: any path from the initial state in $G$ can be mapped to a corresponding path with the same label in $G'$ and vice versa. Next, we show that the graph is Wheeler-sortable:

\begin{lemma}\label{lemma:WG_atomic}
Every node $v$ in $G'$ is atomic.
\end{lemma}

\begin{proof}

If $v$ is in the first block $G'$, then $|P(v)| = 1$, so $v$ is trivially atomic. Suppose $v$ is not in the first block. Then all strings in $P(v)$ are of the form $\alpha \beta \gamma$, with $|\alpha| \geq 0$, $|\beta| > 0$ and $|\gamma| \geq 0$, such that $\beta$ occurs in the NFA $F$ only once (repeat-free property) and $\gamma$ is a prefix of some string of a block. By the construction of $G'$, strings $\beta$ and $\gamma$ are the same for all strings in $P(v)$. 
Consider a string $\delta$ in $G'$ such that $P_{min}(v) \preceq \delta \preceq P_{max}(v)$. Then $\delta$ must be suffixed by $\beta \gamma$, so it follows that $\delta \in P(v)$, since (1) all occurrences of $\beta$ lead to the same node $u$ at the end of the previous block in $G'$ and (2) all paths from $u$ with the label $\gamma$ must lead to $v$, or otherwise $G'$ is not deterministic, a contradiction.
\end{proof}
Next, we show that transformation from $F$ to $G'$ increases the number of nodes by at most a polynomial amount.

\WGsizerestatable*
\begin{proof}
Each non-source node $v$ of $G'$ can be associated to a pair $(u,\alpha)$, where $u$ is the node of $G'$ reached by walking from $v$ back to the end of the previous block, and $\alpha$ is the label of the path from $u$ to $v$. If $v$ is at the first block, we set $u$ to the added source node. If $v$ is at the end of a block, we set $u=v$ and set $\alpha$ to be the empty string. These pairs are all distinct because $G'$ is deterministic. 

We bound the number of nodes in $G'$ by bounding the number of possible distinct pairs $(u,\alpha)$. Each end-of-block node $v$ of $G'$ is paired only with prefixes of strings from the next block. Let $end(b)$ be the set of nodes of $F$ that are at the end of block $b$ and let $f(b)$ be the set of nodes at block $b$ in $F$. Then the total number of possible pairs with nonempty $\alpha$ is at most $\sum_{b = 0}^{B-1} |end(b)| \cdot |f(b+1)|$, where $B$ is the number of blocks, with block zero corresponding to the initial state. The number of pairs with empty $\alpha$ is $\sum_{b = 0}^{B} |end(b)| \leq N$. Bounding $end(b) \leq W$ in the former sum, we have a total of at most $W \sum_{b=0}^{B-1} |f(b+1)| + N = O(WN)$ possible distinct pairs.
\end{proof}

\section{Implementation \label{sect:experiments}}

We implemented construction and indexing of (semi-)repeat-free (elastic) founder graphs. The implementation is available at \url{https://github.com/algbio/founderblockgraphs}. Some proof-of-concept experiments can be found in the conference version of the paper \cite{MCENT20}. 

\section{Discussion\label{sect:discussion}}

One characterization of our solution is that we compact those vertical repeats in \msa that are not horizontal repeats. This can be seen as positional extension of variable order de Bruijn graphs. Also, our solution is parameter-free unlike de Bruijn approaches that always need some threshold $k$, even in the variable order case. 

The founder graph concept could also be generalized so that it is not directly induced from a segmentation. One could consider cyclic graphs having the same repeat-free property. This could be an interesting direction in defining parameter-free de Bruijn graphs.

This paper only scratches the surface of a new family of pangenome representations. There are myriad of options how to optimize among the valid segmentations \cite{NCKM19,CKMN19}. We studied some of these here, but left open how to e.g. minimize the maximum number of distinct strings in a segment (i.e. \emph{height} of the graph) \cite{NCKM19}, or how to control the over-expressiveness of the graph. For the former, our ongoing work gives an efficient solution \cite{RM22b}. 

Other open problems include strengthening the conditional indexing lower bound to cover non-elastic founder graphs, and improving the running time for constructing (semi-)repeat-free elastic founder graphs. For the latter, our ongoing work improves the preprocessing time to linear, and consequently shows that a semi-repeat segmentation maximizing the number of blocks can be computed in linear time \cite{RM22a}.

We focused here on the theoretical aspects of indexable founder graphs. Our preliminary experiments \cite{MCENT20} show that the approach works well in practice on multiple sequence alignments without gaps. In our future work, we will focus on making the approach practical also in the general case. 

\section*{Acknowledgments}

We wish to thank the anonymous reviewers of WABI 2020, RECOMB-seq 2021, and ISAAC 2021 for helping us improve the readability.

This work was partly funded by the Academy of Finland (grants 309048, 322595 and 328877), Helsinki Institute for Information Technology (HIIT), and
by the European Research Council (ERC) under the European Union's Horizon 2020 research and innovation programme (grant agreement No.~851093, SAFEBIO).

\bibliographystyle{plain}

\begin{thebibliography}{10}

\bibitem{AC75}
Alfred~V. Aho and Margaret~J. Corasick.
\newblock Efficient string matching: An aid to bibliographic search.
\newblock {\em Commun. {ACM}}, 18(6):333--340, 1975.

\bibitem{Alaetal20}
Jarno Alanko, Giovanna D'Agostino, Alberto Policriti, and Nicola Prezza.
\newblock Regular languages meet prefix sorting.
\newblock In Shuchi Chawla, editor, {\em Proceedings of the 2020 {ACM-SIAM}
  Symposium on Discrete Algorithms, {SODA} 2020, Salt Lake City, UT, USA,
  January 5-8, 2020}, pages 911--930, {USA}, 2020. {SIAM}.

\bibitem{Alzetal20}
Mai Alzamel, Lorraine A.~K. Ayad, Giulia Bernardini, Roberto Grossi, Costas~S.
  Iliopoulos, Nadia Pisanti, Solon~P. Pissis, and Giovanna Rosone.
\newblock Comparing degenerate strings.
\newblock {\em Fundam. Informaticae}, 175(1-4):41--58, 2020.

\bibitem{AmirLL00}
Amihood Amir, Moshe Lewenstein, and Noa Lewenstein.
\newblock Pattern matching in hypertext.
\newblock {\em J. Algorithms}, 35(1):82--99, 2000.

\bibitem{aoyama_et_alCPM2018}
Kotaro Aoyama, Yuto Nakashima, Tomohiro I, Shunsuke Inenaga, Hideo Bannai, and
  Masayuki Takeda.
\newblock {Faster Online Elastic Degenerate String Matching}.
\newblock In Gonzalo Navarro, David Sankoff, and Binhai Zhu, editors, {\em
  Annual Symposium on Combinatorial Pattern Matching (CPM 2018)}, volume 105 of
  {\em Leibniz International Proceedings in Informatics (LIPIcs)}, pages
  9:1--9:10, Dagstuhl, Germany, 2018. Schloss Dagstuhl--Leibniz-Zentrum fuer
  Informatik.

\bibitem{BC19}
Djamal Belazzougui and Fabio Cunial.
\newblock Fully-functional bidirectional burrows-wheeler indexes and
  infinite-order de bruijn graphs.
\newblock In Nadia Pisanti and Solon~P. Pissis, editors, {\em 30th Annual
  Symposium on Combinatorial Pattern Matching, {CPM} 2019, June 18-20, 2019,
  Pisa, Italy}, volume 128 of {\em LIPIcs}, pages 10:1--10:15, Dagstuhl,
  Germany, 2019. Schloss Dagstuhl - Leibniz-Zentrum f{\"{u}}r Informatik.

\bibitem{belazzougui2013versatile}
Djamal Belazzougui, Fabio Cunial, Juha K{\"a}rkk{\"a}inen, and Veli
  M{\"a}kinen.
\newblock Versatile succinct representations of the bidirectional
  burrows-wheeler transform.
\newblock In {\em European Symposium on Algorithms}, pages 133--144. Springer,
  2013.

\bibitem{BCKM20}
Djamal Belazzougui, Fabio Cunial, Juha K\"{a}rkk\"{a}inen, and Veli
  M\"{a}kinen.
\newblock Linear-time string indexing and analysis in small space.
\newblock {\em ACM Trans. Algorithms}, 16(2):Article 17, March 2020.

\bibitem{bernardini_et_al2019elastic}
Giulia Bernardini, Pawel Gawrychowski, Nadia Pisanti, Solon~P. Pissis, and
  Giovanna Rosone.
\newblock {Even Faster Elastic-Degenerate String Matching via Fast Matrix
  Multiplication}.
\newblock In Christel Baier, Ioannis Chatzigiannakis, Paola Flocchini, and
  Stefano Leonardi, editors, {\em 46th International Colloquium on Automata,
  Languages, and Programming (ICALP 2019)}, volume 132 of {\em Leibniz
  International Proceedings in Informatics (LIPIcs)}, pages 21:1--21:15,
  Dagstuhl, Germany, 2019. Schloss Dagstuhl--Leibniz-Zentrum fuer Informatik.

\bibitem{Beretal17}
Giulia Bernardini, Nadia Pisanti, Solon~P. Pissis, and Giovanna Rosone.
\newblock Pattern matching on elastic-degenerate text with errors.
\newblock In Gabriele Fici, Marinella Sciortino, and Rossano Venturini,
  editors, {\em String Processing and Information Retrieval - 24th
  International Symposium, {SPIRE} 2017, Palermo, Italy, September 26-29, 2017,
  Proceedings}, volume 10508 of {\em Lecture Notes in Computer Science}, pages
  74--90, Germany, 2017. Springer.

\bibitem{Beretal20}
Giulia Bernardini, Nadia Pisanti, Solon~P. Pissis, and Giovanna Rosone.
\newblock Approximate pattern matching on elastic-degenerate text.
\newblock {\em Theor. Comput. Sci.}, 812:109--122, 2020.

\bibitem{BW94}
M.~Burrows and D.~Wheeler.
\newblock A block-sorting lossless data compression algorithm.
\newblock Technical Report 124, Digital Equipment Corporation, 1994.

\bibitem{CKMN19}
Bastien Cazaux, Dmitry Kosolobov, Veli M{\"{a}}kinen, and Tuukka Norri.
\newblock Linear time maximum segmentation problems in column stream model.
\newblock In Nieves~R. Brisaboa and Simon~J. Puglisi, editors, {\em String
  Processing and Information Retrieval - 26th International Symposium, {SPIRE}
  2019, Segovia, Spain, October 7-9, 2019, Proceedings}, volume 11811 of {\em
  Lecture Notes in Computer Science}, pages 322--336, Germany, 2019. Springer.

\bibitem{Chaetal15}
Maria Chatzou, Cedrik Magis, Jia-Ming Chang, Carsten Kemena, Giovanni Bussotti,
  Ionas Erb, and Cedric Notredame.
\newblock {Multiple sequence alignment modeling: methods and applications}.
\newblock {\em Briefings in Bioinformatics}, 17(6):1009--1023, 11 2015.

\bibitem{cunial2019framework}
Fabio Cunial, Jarno Alanko, and Djamal Belazzougui.
\newblock A framework for space-efficient variable-order markov models.
\newblock {\em Bioinformatics}, 35(22):4607--4616, 2019.

\bibitem{Bri59}
Rene De~La~Briandais.
\newblock File searching using variable length keys.
\newblock In {\em Papers Presented at the the March 3-5, 1959, Western Joint
  Computer Conference}, IRE-AIEE-ACM â€™59 (Western), page 295â€“298, New York,
  NY, USA, 1959. Association for Computing Machinery.

\bibitem{EGMT19}
Massimo Equi, Roberto Grossi, Veli M{\"{a}}kinen, and Alexandru~I. Tomescu.
\newblock On the complexity of string matching for graphs.
\newblock In Christel Baier, Ioannis Chatzigiannakis, Paola Flocchini, and
  Stefano Leonardi, editors, {\em 46th International Colloquium on Automata,
  Languages, and Programming, {ICALP} 2019, July 9-12, 2019, Patras, Greece},
  volume 132 of {\em LIPIcs}, pages 55:1--55:15, Dagstuhl, Germany, 2019.
  Schloss Dagstuhl - Leibniz-Zentrum f{\"{u}}r Informatik.

\bibitem{EMT21}
Massimo Equi, Veli M{\"{a}}kinen, and Alexandru~I. Tomescu.
\newblock Graphs cannot be indexed in polynomial time for sub-quadratic time
  string matching, unless {SETH} fails.
\newblock In Tom{\'{a}}s Bures, Riccardo Dondi, Johann Gamper, Giovanna
  Guerrini, Tomasz Jurdzinski, Claus Pahl, Florian Sikora, and Prudence W.~H.
  Wong, editors, {\em {SOFSEM} 2021: Theory and Practice of Computer Science -
  47th International Conference on Current Trends in Theory and Practice of
  Computer Science, {SOFSEM} 2021, Bolzano-Bozen, Italy, January 25-29, 2021,
  Proceedings}, volume 12607 of {\em Lecture Notes in Computer Science}, pages
  608--622, Germany, 2021. Springer.

\bibitem{ENACTM21}
Massimo Equi, Tuukka Norri, Jarno Alanko, Bastien Cazaux, Alexandru~I. Tomescu,
  and Veli M{\"{a}}kinen.
\newblock Algorithms and complexity on indexing elastic founder graphs.
\newblock In Hee{-}Kap Ahn and Kunihiko Sadakane, editors, {\em 32nd
  International Symposium on Algorithms and Computation, {ISAAC} 2021, December
  6-8, 2021, Fukuoka, Japan}, volume 212 of {\em LIPIcs}, Dagstuhl, Germany,
  2021. Schloss Dagstuhl - Leibniz-Zentrum f{\"{u}}r Informatik.
\newblock pp. 20:1--20:18.

\bibitem{Farach97}
Martin Farach.
\newblock Optimal suffix tree construction with large alphabets.
\newblock In {\em Proceedings 38th Annual Symposium on Foundations of Computer
  Science}, pages 137--143. IEEE, 1997.

\bibitem{GBT84}
Harold~N. Gabow, Jon~Louis Bentley, and Robert~Endre Tarjan.
\newblock Scaling and related techniques for geometry problems.
\newblock In Richard~A. DeMillo, editor, {\em Proceedings of the 16th Annual
  {ACM} Symposium on Theory of Computing, April 30 - May 2, 1984, Washington,
  DC, {USA}}, pages 135--143, {USA}, 1984. {ACM}.

\bibitem{GMS17}
Travis Gagie, Giovanni Manzini, and Jouni Sir{\'{e}}n.
\newblock Wheeler graphs: {A} framework for bwt-based data structures.
\newblock {\em Theor. Comput. Sci.}, 698:67--78, 2017.

\bibitem{GN19}
Travis Gagie and Gonzalo Navarro.
\newblock Compressed indexes for repetitive textual datasets.
\newblock In Sherif Sakr and Albert~Y. Zomaya, editors, {\em Encyclopedia of
  Big Data Technologies}. Springer, Germany, 2019.

\bibitem{GNP20}
Travis Gagie, Gonzalo Navarro, and Nicola Prezza.
\newblock Fully functional suffix trees and optimal text searching in bwt-runs
  bounded space.
\newblock {\em J. {ACM}}, 67(1):2:1--2:54, 2020.

\bibitem{GibneySPIRE2020}
Daniel Gibney.
\newblock An efficient elastic-degenerate text index? not likely.
\newblock In {\em International Symposium on String Processing and Information
  Retrieval}, pages 76--88. Springer, 2020.

\bibitem{GT19}
Daniel Gibney and Sharma~V. Thankachan.
\newblock On the hardness and inapproximability of recognizing wheeler graphs.
\newblock In Michael~A. Bender, Ola Svensson, and Grzegorz Herman, editors,
  {\em 27th Annual European Symposium on Algorithms, {ESA} 2019, September
  9-11, 2019, Munich/Garching, Germany}, volume 144 of {\em LIPIcs}, pages
  51:1--51:16, Germany, 2019. Schloss Dagstuhl - Leibniz-Zentrum f{\"{u}}r
  Informatik.

\bibitem{GGV03}
Roberto Grossi, Ankur Gupta, and Jeffrey~Scott Vitter.
\newblock High-order entropy-compressed text indexes.
\newblock In {\em Proceedings of the Fourteenth Annual {ACM-SIAM} Symposium on
  Discrete Algorithms, January 12-14, 2003, Baltimore, Maryland, {USA}}, pages
  841--850, {USA}, 2003. {ACM/SIAM}.

\bibitem{IKP17}
Costas~S. Iliopoulos, Ritu Kundu, and Solon~P. Pissis.
\newblock Efficient pattern matching in elastic-degenerate texts.
\newblock In Frank Drewes, Carlos Mart{\'{\i}}n{-}Vide, and Bianca Truthe,
  editors, {\em Language and Automata Theory and Applications - 11th
  International Conference, {LATA} 2017, Ume{\aa}, Sweden, March 6-9, 2017,
  Proceedings}, volume 10168 of {\em Lecture Notes in Computer Science}, pages
  131--142, 2017.

\bibitem{CR16}
Costas~S. Iliopoulos and Jakub Radoszewski.
\newblock Truly subquadratic-time extension queries and periodicity detection
  in strings with uncertainties.
\newblock In Roberto Grossi and Moshe Lewenstein, editors, {\em 27th Annual
  Symposium on Combinatorial Pattern Matching, {CPM} 2016, June 27-29, 2016,
  Tel Aviv, Israel}, volume~54 of {\em LIPIcs}, pages 8:1--8:12, Dagstuhl,
  Germany, 2016. Schloss Dagstuhl - Leibniz-Zentrum f{\"{u}}r Informatik.

\bibitem{IP01}
Russell Impagliazzo and Ramamohan Paturi.
\newblock {On the Complexity of k-SAT}.
\newblock {\em Journal of Computer and System Sciences}, 62(2):367 -- 375,
  2001.

\bibitem{Jac89}
G.~Jacobson.
\newblock Space-efficient static trees and graphs.
\newblock In {\em Proc. FOCS}, pages 549--554, 1989.

\bibitem{Mai78}
David Maier.
\newblock The complexity of some problems on subsequences and supersequences.
\newblock {\em J. ACM}, 25(2):322â€“336, April 1978.

\bibitem{MCENT20}
Veli M{\"{a}}kinen, Bastien Cazaux, Massimo Equi, Tuukka Norri, and
  Alexandru~I. Tomescu.
\newblock Linear time construction of indexable founder block graphs.
\newblock In Carl Kingsford and Nadia Pisanti, editors, {\em 20th International
  Workshop on Algorithms in Bioinformatics, {WABI} 2020, September 7-9, 2020,
  Pisa, Italy (Virtual Conference)}, volume 172 of {\em LIPIcs}, pages
  7:1--7:18, Dagstuhl, Germany, 2020. Schloss Dagstuhl - Leibniz-Zentrum
  f{\"{u}}r Informatik.

\bibitem{MNSV09jcb}
Veli M{\"a}kinen, Gonzalo Navarro, Jouni Sir{\'e}n, and Niko V{\"a}lim{\"a}ki.
\newblock Storage and retrieval of highly repetitive sequence collections.
\newblock {\em Journal of Computational Biology}, 17(3):281--308, 2010.

\bibitem{manber1992approximate}
U.~Manber and S.~Wu.
\newblock Approximate string matching with arbitrary costs for text and
  hypertext.
\newblock In {\em IAPR Workshop on Structural and Syntactic Pattern
  Recognition, Bern, Switzerland}, pages 22--33, 1992.

\bibitem{MM93}
Udi Manber and Eugene~W. Myers.
\newblock Suffix arrays: {A} new method for on-line string searches.
\newblock {\em {SIAM} J. Comput.}, 22(5):935--948, 1993.

\bibitem{marschall2016computational}
Tobias Marschall, Manja Marz, Thomas Abeel, Louis Dijkstra, Bas~E Dutilh, Ali
  Ghaffaari, Paul Kersey, Wigard Kloosterman, Veli M{\"a}kinen, Adam Novak,
  et~al.
\newblock Computational pan-genomics: status, promises and challenges.
\newblock {\em BioRxiv}, page 043430, 2016.

\bibitem{Naetal16b}
Joong Na, Hyunjoon Kim, Seunghwan Min, Heejin Park, Thierry Lecroq, Martine
  Leonard, Laurent Mouchard, and Kunsoo Park.
\newblock {FM}-index of alignment with gaps.
\newblock {\em Theoretical Computer Science}, 710, 06 2016.

\bibitem{Naetal16a}
Joong~Chae Na, Hyunjoon Kim, Heejin Park, Thierry Lecroq, Martine LÃ©onard,
  Laurent Mouchard, and Kunsoo Park.
\newblock {FM}-index of alignment: A compressed index for similar strings.
\newblock {\em Theoretical Computer Science}, 638:159--170, 2016.
\newblock Pattern Matching, Text Data Structures and Compression.

\bibitem{Naetal13a}
Joong~Chae Na, Heejin Park, Maxime Crochemore, Jan Holub, Costas~S. Iliopoulos,
  Laurent Mouchard, and Kunsoo Park.
\newblock Suffix tree of alignment: An efficient index for similar data.
\newblock In Thierry Lecroq and Laurent Mouchard, editors, {\em Combinatorial
  Algorithms - 24th International Workshop, {IWOCA} 2013, Rouen, France, July
  10-12, 2013, Revised Selected Papers}, volume 8288 of {\em Lecture Notes in
  Computer Science}, pages 337--348, Germany, 2013. Springer.

\bibitem{Naetal13b}
Joong~Chae Na, Heejin Park, Sunho Lee, Minsung Hong, Thierry Lecroq, Laurent
  Mouchard, and Kunsoo Park.
\newblock Suffix array of alignment: {A} practical index for similar data.
\newblock In Oren Kurland, Moshe Lewenstein, and Ely Porat, editors, {\em
  String Processing and Information Retrieval - 20th International Symposium,
  {SPIRE} 2013, Jerusalem, Israel, October 7-9, 2013, Proceedings}, volume 8214
  of {\em Lecture Notes in Computer Science}, pages 243--254, Germany, 2013.
  Springer.

\bibitem{NCKM19}
Tuukka Norri, Bastien Cazaux, Dmitry Kosolobov, and Veli M{\"{a}}kinen.
\newblock Linear time minimum segmentation enables scalable founder
  reconstruction.
\newblock {\em Algorithms Mol. Biol.}, 14(1):12:1--12:15, 2019.

\bibitem{rautiainen2017aligning}
Mikko Rautiainen and Tobias Marschall.
\newblock Aligning sequences to general graphs in {$O(V+ mE)$} time.
\newblock {\em bioRxiv}, pages 216--127, 2017.

\bibitem{RM22b}
Nicola Rizzo and Veli M\"akinen.
\newblock Indexable elastic founder graphs of minimum height.
\newblock In {\em Proc. 33rd Annual Symposium on Combinatorial Pattern Matching
  (CPM 2022)}, 2022.
\newblock To appear.

\bibitem{RM22a}
Nicola Rizzo and Veli M{\"{a}}kinen.
\newblock Linear time construction of indexable elastic founder graphs.
\newblock In Cristina Bazgan and Henning Fernau, editors, {\em Combinatorial
  Algorithms - 33rd International Workshop, {IWOCA} 2022, Trier, Germany, June
  7-9, 2022, Proceedings}, volume 13270 of {\em Lecture Notes in Computer
  Science}, pages 480--493. Springer, 2022.

\bibitem{Sad07}
Kunihiko Sadakane.
\newblock Compressed suffix trees with full functionality.
\newblock {\em Theory Comput. Syst.}, 41(4):589--607, 2007.

\bibitem{SOG12}
Thomas Schnattinger, Enno Ohlebusch, and Simon Gog.
\newblock Bidirectional search in a string with wavelet trees and bidirectional
  matching statistics.
\newblock {\em Inf. Comput.}, 213:13--22, 2012.

\bibitem{SVM14}
Jouni Sir{\'e}n, Niko V{\"a}lim{\"a}ki, and Veli M{\"a}kinen.
\newblock Indexing graphs for path queries with applications in genome
  research.
\newblock {\em IEEE/ACM Transactions on Computational Biology and
  Bioinformatics}, 11(2):375--388, 2014.

\bibitem{Tha13}
Chris Thachuk.
\newblock Indexing hypertext.
\newblock {\em Journal of Discrete Algorithms}, 18:113 -- 122, 2013.
\newblock Selected papers from the 18th International Symposium on String
  Processing and Information Retrieval (SPIRE 2011).

\bibitem{Williams05}
Ryan Williams.
\newblock A new algorithm for optimal 2-constraint satisfaction and its
  implications.
\newblock {\em Theor. Comput. Sci.}, 348(2-3):357--365, 2005.

\end{thebibliography}

\end{document}